%% file: regular-paper.tex
\documentclass[10pt,journal,compsoc]{IEEEtran}
\usepackage{amsthm,amsmath,amsfonts,amssymb}
\usepackage{algorithmic}
\usepackage{algorithm}
\usepackage{array}
\usepackage[caption=false,font=normalsize,labelfont=sf,textfont=sf]{subfig}
\usepackage{textcomp}
\usepackage{stfloats}
\usepackage{url}
\usepackage{verbatim}
\usepackage{graphicx}
\usepackage{cite}
\usepackage{dsfont}
\usepackage{enumitem}
\usepackage{multirow}
\usepackage{caption}
\usepackage{threeparttable}
\usepackage[table,xcdraw]{xcolor}
\usepackage{bbding}
\usepackage{booktabs}
\usepackage{pifont}
\usepackage[latin1]{inputenc}

\allowdisplaybreaks[4]

\graphicspath{{figures/},{figures/photos/},{figures/experiments/}\:}
\DeclareGraphicsExtensions{.pdf,.jpg,.png}

\floatstyle{boxed}
\newfloat{fig}{bt}{fig}
\floatname{fig}{Fig.}
\AtBeginDocument{\captionsetup[fig]{labelsep=period}}
\newtheorem{definition}{\textbf{Definition}}
\newtheorem{theorem}{\textbf{Theorem}}

\usepackage{xspace}
\def\nameScheme{ShieldFL\xspace}

\usepackage[colorlinks=true,
            linkcolor=black, %%%ÐÞ¸Ä´Ë´¦ÎªÄãÏëÒªµÄÑÕÉ«
            anchorcolor=black, %%ÐÞ¸Ä´Ë´¦ÎªÄãÏëÒªµÄÑÕÉ«
            citecolor=black, %%ÉèÖÃÎÄÖÐÒýÓÃ²Î¿¼ÎÄÏ×µÄÑÕÉ«
            urlcolor=blue%%ÉèÖÃDOIµÄÑÕÉ«ÎªÀ¶É«
            ]{hyperref}

\begin{document}
\title{Privacy-Preserving Federated Learning Scheme with Mitigating Model Poisoning Attacks: Vulnerabilities and Countermeasures}
%\vspace*{-1\baselineskip}
\author{Jiahui~Wu, Fucai~Luo, Tiecheng~Sun,~\IEEEmembership{Member,~IEEE}, Haiyan~Wang, Weizhe~Zhang,~\IEEEmembership{Senior Member,~IEEE}
        % <-this % stops a space
\thanks{This work is supported by the Key Program of the Joint Fund of the National Natural Science Foundation of China (Grant No. U22A2036) and the Major Key Project of Peng Cheng Laboratory (Grant No. PCL2023A06).}% <-this % stops a space
%\thanks{Jiahui Wu and Weizhe Zhang are with the Department of New Networks, Peng Cheng Laboratory, Shenzhen 518000, China (e-mail: wjh01@pcl.ac.cn; weizhe.zhang@pcl.ac.cn);
\thanks{J. Wu, T. Sun and H. Wang are with the New Network Department, Peng Cheng Laboratory, Shenzhen 518000, China; F. Luo is with the School of Computer Science and Technology, Zhejiang Gongshang University, Hangzhou 310018, China;
W. Zhang is with the School of Cyberspace Science Faculty of Computing, Harbin Institute of Technology, Shenzhen 518055, China and with the New Network Department, Peng Cheng Laboratory, Shenzhen 518000, China.
E-mail: wujh01@pcl.ac.cn; lfucai@126.com; tiechengsun@126.com; wanghy01@pcl.ac.cn; wzzhang@hit.edu.cn.
Corresponding author: W. Zhang.}}

%\maketitle

\IEEEtitleabstractindextext{

}

\maketitle

\input{sections/1-introduction}

\input{sections/1-related}

\input{sections/2-problem}

\input{sections/2-vulnerabilities}

\input{sections/3-analysis}

\input{sections/4-enhanced}

\input{sections/6-conclusion}

  \bibliographystyle{IEEEtran}
   %\scriptsize
  \bibliography{cited,IEEEabrv} %²Î¿¼ÎÄÏ×ÖÐµÄÆÚ¿¯ÃûÊ¹ÓÃIEEEabrv.bibÖÐ¶¨ÒåµÄÆÚ¿¯ÃûËõÐ´
%\end{spacing}

\end{document}

%% file: sections/1-introduction.tex
\section{Introduction}

Federated learning (FL) is a new distributed learning paradigm, which enables collaborative model training among a group of distributed users and a central aggregation server, while keeping users' sensitive data local.  In an ideal scenario free from distrustful entities, FL empowers users to retain complete control over their data, thereby facilitating the development of high-quality machine learning models. However, in practice, both servers and users may lack absolute trustworthiness, leading to various attacks
\cite{hitaj2017deep,geiping2020inverting,yang2024fast,bhagoji2019analyzing,shi2024towards,li2022learning}
against the FL framework. In recent years, the US government has been actively promoting research aimed at creating robust and privacy-preserving FL solutions for diverse tasks
such as financial crime prevention \cite{PETsPrizeForFinancial} and pandemic response and forecasting \cite{PETsPrizeForPandemic}.
Incentives, including substantial prizes totaling $\$800,000$, have been offered to spur innovation in this area.
Over the years, extensive researches \cite{Chen2024Differentially,hu2024maskcrypt,ye2022one,chang2023privacy,zheng2022aggregation,2017Machine,2021FLTrust,krauss2024automatic,yan2024recess,lu2024depriving} have underscored the importance of two critical security properties in FL:
\begin{itemize}[leftmargin=0.3cm]
  \item \textit{Privacy protection of gradients:}
   In FL, the central server may recover sensitive information of the users' training datasets by analyzing the local gradients received from users
   \cite{hitaj2017deep,geiping2020inverting,yang2024fast}.
   To resist privacy leaks from gradients, privacy-preserving federated learning (PPFL) techniques  \cite{Chen2024Differentially,hu2024maskcrypt,ye2022one,chang2023privacy,zheng2022aggregation} are developed. These techniques involve users masking their local gradients using privacy protection methods before submitting them to the server for secure aggregation.

  \item \textit{Mitigation of poisoning gradients: }
    Malicious users may be manipulated by Byzantine adversaries to send poisonous local gradients, deviating from correct model training.
    The adversarial objective may involve forcing the model to classify inputs into specific incorrect classes (targeted attacks) \cite{bhagoji2019analyzing,shi2024towards} or inducing misclassifications without specifying a particular target class (untargeted attacks) \cite{li2022learning}. Such model poisoning attacks can severely compromise the model's performance. Recent works \cite{2021FLTrust,krauss2024automatic,yan2024recess,lu2024depriving}
    show that cosine similarity is an effective method to detect poisonous gradients.
\end{itemize}

The existing solutions for providing privacy and defending against poisoning attacks in FL focus on two opposing directions: privacy-preserving FL solutions aim to achieve data indistinguishability, whereas defenses against distorted gradients often rely on the similarity between malicious and benign gradients to filter out malicious ones. Specifically, defenses that extract valuable statistical information largely depend on the distinguishability of the data. Therefore, simultaneously achieving both security goals in FL remains challenging.
Recently, some efforts have been made to simultaneously achieve both objectives.
For instance,  researchers propose privacy-preserving and Byzantine-robust FL (PBFL) schemes.
However, these schemes often suffer from inefficiency \cite{lin2022ppbr,Robust2024Hao,ma2021pocket,kasyap2022efficient,xu2022mudfl,Abdel2022Privacy}, still make privacy leakages \cite{kasyap2022efficient,xu2022mudfl,Abdel2022Privacy,liu2021privacy, ma2022shieldfl,zhang2022lsfl}, or not secure for secure multiparty computation (MPC) scenarios like PPFL applications \cite{miao2022privacy,miao2024rfed}.
Based on server aggregation settings, we classify existing PBFL schemes into three categories: the serverless/decentralized model, the single-server model, and the two-server model.
Table~\ref{tab:relatedWork} presents a qualitative comparison of FL representative works with different server aggregation settings.
\begin{table}[]
%\small
\footnotesize
\setlength{\tabcolsep}{4pt}
 \centering
\caption{\small{Qualitative comparison of FL representative works.}}\label{tab:relatedWork}
\begin{tabular}{|c||c|c|c|c|}
\hline
\textbf{Schemes} & \begin{tabular}[c]{@{}c@{}}\textbf{Privacy}\\ \textbf{Protection}\end{tabular} & \begin{tabular}[c]{@{}c@{}}\textbf{Poisoning}\\ \textbf{Resilience}\end{tabular} & \begin{tabular}[c]{@{}c@{}}\textbf{Server}\\ \textbf{Setting}\end{tabular} & \textbf{Efficiency}                                          \\ \hline\hline
\textbf{FedSGD} \cite{mcmahan2017communication}               & \cellcolor[HTML]{FD6864}\XSolidBrush          & \cellcolor[HTML]{FD6864}\XSolidBrush            & \cellcolor[HTML]{FD6864}$1$                    & \cellcolor[HTML]{32CB00}\Checkmark   \\ \hline
\textbf{PrivateFL} \cite{Chen2024Differentially,hu2024maskcrypt,ye2022one,chang2023privacy,zheng2022aggregation} & \cellcolor[HTML]{32CB00}\Checkmark            & \cellcolor[HTML]{FD6864}\XSolidBrush            & \cellcolor[HTML]{FD6864}$0/1/2$                                & \cellcolor[HTML]{32CB00}\Checkmark   \\ \hline
\textbf{RobustFL} \cite{2017Machine,2021FLTrust,krauss2024automatic,yan2024recess,lu2024depriving}
                 & \cellcolor[HTML]{FD6864}\XSolidBrush          & \cellcolor[HTML]{32CB00}\Checkmark              & \cellcolor[HTML]{FD6864}$0/1/2$                                & \cellcolor[HTML]{32CB00}\Checkmark   \\ \hline
\textbf{PBFL} \cite{kasyap2022efficient,xu2022mudfl,Abdel2022Privacy} & \cellcolor[HTML]{FD6864}\XSolidBrush          & \cellcolor[HTML]{32CB00}\Checkmark              & \cellcolor[HTML]{FD6864}$0$                      & \cellcolor[HTML]{FD6864}\XSolidBrush \\ \hline
\textbf{PBFL} \cite{lin2022ppbr,Robust2024Hao,ma2021pocket}& \cellcolor[HTML]{32CB00}\Checkmark            & \cellcolor[HTML]{32CB00}\Checkmark              & \cellcolor[HTML]{FD6864}$1$                    & \cellcolor[HTML]{FD6864}\XSolidBrush \\ \hline
\textbf{PBFL} \cite{liu2021privacy, ma2022shieldfl,zhang2022lsfl,miao2022privacy,miao2024rfed}& \cellcolor[HTML]{FD6864}\XSolidBrush          & \cellcolor[HTML]{32CB00}\Checkmark              & \cellcolor[HTML]{32CB00}$2$                      & \cellcolor[HTML]{32CB00}\Checkmark   \\ \hline
\rowcolor[HTML]{32CB00}
\cellcolor[HTML]{FFFFFF}\textbf{Our scheme  }               & \Checkmark                                    & \Checkmark                                      & \cellcolor[HTML]{32CB00}$2$                      & \Checkmark                           \\ \hline
\end{tabular}
\begin{tablenotes}
  \footnotesize
  %\scriptsize
   \item \textbf{Notes.} In the column of ``Server Setting", $0,1,$ and $2$ represent the server-less, the single-server, and the two-server models, respectively.
\end{tablenotes}
%\vspace*{-2\baselineskip}
\end{table}

\textbf{In the serverless/decentralized model}, multiple users collaborate in training a model without any central server.
Blockchain techniques are commonly employed to resist model poisoning attacks in FL with the serverless setting (e.g., \cite{kasyap2022efficient,xu2022mudfl,Abdel2022Privacy}) since it provides a distributed and transparent framework for storing and managing model updates in FL, thereby reducing the risk of model poisoning attacks. However, these schemes bring additional latency due to the consensus process and result in exponentially increasing communication costs among clients, thereby diminishing the efficiency of FL. Moreover, these schemes encounter challenges in privacy protection because of the openness and transparency of the data stored on the blockchain, and the introduction of privacy protection technology will bring higher storage costs and may be limited by low transaction throughput.
\textbf{FL schemes with the single-server model} \cite{lin2022ppbr,Robust2024Hao,ma2021pocket} involve a central server in the gradient aggregation process.
These schemes typically employ privacy protection techniques based on multiple keys/parameters to ensure reliable privacy protection among users, such as local sensitivities of local differential privacy in \cite{lin2022ppbr} and individual keys of multi-key HE in \cite{Robust2024Hao,ma2021pocket}. Despite offering privacy protection, these schemes are often inefficient. Furthermore,
relying on a single server exposes it to potential single points of failure, leaving it vulnerable to attacks from external adversaries. These attacks could result in breaches of model privacy through data infringements, hacking, or leaks.

\textbf{FL schemes with the two-server model} (e.g., \cite{liu2021privacy, ma2022shieldfl, zhang2022lsfl, miao2022privacy, miao2024rfed}) distribute the trust assumed in the traditional single-server model across two non-colluding servers. These schemes are promising as they eliminate single points of failure and enhance robustness.
However, even state-of-the-art schemes such as ShieldFL contain a critical privacy vulnerability that can completely undermine their guarantees: the robust aggregation process leaks the relative relationships between clients' gradients. A single compromised client can leverage this leakage to reconstruct other clients' gradients in clear.
In ShieldFL \cite{ma2022shieldfl}, each client encrypts its gradient under a unique key, and two servers jointly perform decryption and aggregation. To ensure no server directly learns the gradients, ShieldFL masks each gradient with random noise before cosine similarity is computed for robust aggregation. However, the same noise is applied to all gradients and to the baseline gradient. This design preserves correctness after noise removal, but it also inadvertently reveals the relative offsets between gradients. If one client is malicious, it can subtract its own gradient and the shared noise from the noisy similarity results to infer the exact gradients of all other clients, thereby breaching the privacy of the entire system.
This example highlights a fundamental tension in existing two-server FL schemes: the need to preserve \textbf{computational correctness} (e.g., accurate robust aggregation) often conflicts with \textbf{data confidentiality}.
Recent works such as \cite{miao2022privacy, miao2024rfed} employ single-key homomorphic encryption (HE) to protect the confidentiality of data from multiple clients. While this approach is effective in some cases, it fails to guarantee confidentiality \textit{between} clients because all ciphertexts are encrypted under the same key. Conversely, schemes based on \textbf{multi-key HE} can protect inter-client confidentiality but face the challenge of performing correct computations across multiple keys.
For example, existing multi-key HE-based schemes like \cite{liu2021privacy, ma2022shieldfl, zhang2022lsfl}, while offering strong privacy protection, may compromise it in order to ensure computational correctness during model aggregation and poisoning defense.
In ShieldFL, although encryption and noise masking are employed to hide gradients, the robust aggregation mechanism still leaks information as explained above. This vulnerability means that once a client is compromised, the system¡¯s privacy guarantees collapse.

Therefore, despite advances in secure computation, existing two-server FL schemes still struggle to simultaneously achieve strong data confidentiality and computational correctness in the presence of model poisoning attacks.

In light of these challenges, this paper identifies weaknesses in PBFL schemes with the two-server model and proposes a solution.
The main contributions are threefold:

  %\item
  (1) We analyze the security of the existing multi-key HE-based PBFL scheme based on the two-server setting and show that it violates privacy.
  Concretely, we present that the main secure computation procedure of the existing two-server-based scheme leaks significant information about the gradients to a semi-honest server.
  This vulnerability may allow the server to obtain clear gradients from all users, potentially compromising the confidentiality/privacy of the poisoning defense process.

  (2) We propose an enhanced scheme based on the two-server model to protect data privacy while bolstering defenses against model poisoning attacks.
  In our scheme,
  we devise a Byzantine-tolerant aggregation method to fortify against model poisoning attacks.
  Additionally, we develop an enhanced secure normalization judgment $\mathbf{ESecJudge}$ method and an enhanced secure cosine similarity measurement $\mathbf{ESecCos}$ method
  to protect the privacy of the defense process while ensuring the correctness of the computations.
  Our scheme guarantees privacy preservation and resilience against model poisoning attacks,
  even in scenarios with heterogeneous datasets that are non-IID (Independently Identically Distributed).
      Furthermore, it exhibits superior computational and communication efficiency.

  (3) Through experiments, we demonstrate the efficacy of our privacy attacks on the existing two-server-based scheme, the effectiveness of our defense against model poisoning attacks, and the efficiency of our scheme in terms of computational and communication costs.
      The experiment results also show that our scheme outperforms the state-of-the-art PBFL schemes.

The remainder of this paper is organized as follows.
In Section~\ref{sec:related}, we introduce the related PBFL schemes.
In Section~\ref{sec:ProblemFormulation}, we introduce the system model, the threat model, and the design goals.
In Section~\ref{sec:oriScheme}, we overview the two-server-based PBFL scheme \nameScheme and present our privacy attacks to \nameScheme.
To defeat our attacks and achieve design goals, an enhanced scheme is proposed in Section~\ref{sec:enhanced}.
Section~\ref{sec:theorem} analyzes the security and the complexity of the enhanced scheme.
%Section~\ref{sec:experiment} introduces the experiments to verify the effectiveness of our attack on \nameScheme and the effectiveness of our enhanced scheme against model poisoning attacks and efficiency metrics.
Finally, conclusions are given in Section~\ref{sec:conclusion}. We give the necessary notations in Table~\ref{tab:notations}.

\begin{table}[t]
 %\footnotesize
\scriptsize
\setlength{\tabcolsep}{2pt}
 \centering
 \renewcommand\arraystretch{1.23}
 \caption{The notations and their semantic meanings.}\label{tab:notations}
 \begin{tabular}{m{0.7cm}|m{3.3cm}||m{0.7cm}|m{3.3cm}}
 %\begin{tabular}{m{1.3cm}|m{6.6cm}}
  \hline
  \textbf{Nots.}   & \textbf{Meanings}& \textbf{Nots.}   & \textbf{Meanings}\\\hline \hline
%  $n$ & The number of the users & $T$ & The number of the iterations\\
%  \hline
  $\epsilon,\epsilon'$ & Security parameters & $p_1,p_2$ & Two odd primes\\
  \hline
  $p, q$ & Plaintext/ciphertext modulus & $P$ & Special modulus\\
  \hline
  $n_f$ & Ring dimension & $\mathcal{R}_q$ & Cyclotomic ring\\
  \hline
  $\mathbb{R}, \mathbb{C}$ & Real/complex number field & $\chi_s,\chi_e$ & Probability distribution\\ %on $\mathcal{R}$\\
  \hline
  $pk,sk,$ $evk$ & Public/secret/evaluation key & $sk_i$ & Secret key shares\\
  \hline
  $[\![\cdot]\!],$ $[\cdot]_{i}$ & Ciphertext/partially decrypted ciphertext & $W^{(t)}$ & Model parameter, a flattened vector of all weight matrixes in a learning model\\
  \hline
  $\|\cdot\|$ & $l_2$ norm of a vector & $\cos$ & Cosine similarity\\
  \hline
  $co_i,$ $cs_i$ & Confidence/credit score of $U_i$ & $\eta$ & Learning rate of gradient descent method\\%
  \hline
  $g_i^{(t)},$ $g_{\Sigma}^{(t)}$ & Local/global gradient vector & $g_{*}^{(t)}$ &Poisonous gradient vector\\ %$g_{i*}^{(t)},$ $g_{*}^{(t)}$ &Poisonous gradient vector/baseline\\
  %$g_i^{(t)}$ & Local gradient vector of the $i$th user in the $t$th iteration\\%; $g_i^{(t)}=\{x_{i1},\cdots,x_{il}\}$ \\
  \hline
  $\alpha$ & Parameter of Dirichlet distribution & $Att_{ro}$ & Attack ratio: the proportion of malicious users to the total number of users\\
  \hline
 \end{tabular}%}}
%\vspace*{-1.5\baselineskip}
\end{table}

%% file: sections/1-related.tex
\vspace*{-0.5\baselineskip}

\section{Related Work}\label{sec:related}

\textbf{FL with the server-less/decentralized model.} 
Blockchain techniques are usually employed to defend against poisoning attacks in this model 
since it provides a distributed and transparent framework for storing and managing model updates in FL, thus reducing the risk of such attacks.
Kasyap et al. \cite{kasyap2022efficient} propose a new Blockchain consensus mechanism based on honest users' contributions on FL training,
incentivizing positive training behaviors and thereby reducing the likelihood of model poisoning attacks.
Xu et al. \cite{xu2022mudfl} propose recording an authenticator of the training program and global model on a Blockchain to resist model poisoning attacks.
This scheme allows FL users to verify the global models they receive by querying the authenticator.
Meanwhile, Mohamed et al. \cite{Abdel2022Privacy} design an incentive strategy that rewards reputation values to each participating user, records these values on a blockchain, and empowers users to choose users with high historical reputations for model aggregation.
These schemes preserve the integrity of the federated model against poisoning attacks, but they
introduce additional latency from the consensus process and incur exponentially increasing communication costs between the users, thereby reducing FL efficiency.
Moreover, these schemes still face privacy protection challenges due to the openness and transparency of blockchain-stored data. 
Introducing privacy protection technology may escalate storage costs and could be hindered by low transaction throughput.

\textbf{FL with single-server model.}
In general, existing researches on the single-server model achieve the secure computations of Byzantine attack detection and model aggregation by directly combining general privacy-enhancing techniques or infrastructures such as differential privacy (DP) \cite{rathee2023elsa,lin2022ppbr}, 
homomorphic encryption (HE) \cite{tang2023pile,Robust2024Hao,ma2021pocket}, secure multiparty computation (MPC) \cite{roy2022eiffel,zhang2023safelearning}, 
and hardware-based trusted execution environment (TEE) \cite{zhang2023agrevader}.
In these schemes, the different users masked their local gradients using either a unique parameter/key  (e.g., a sensitivity of global DP in \cite{rathee2023elsa} and a secret key of single-key HE in \cite{tang2023pile}) or using their individual parameters/keys  (e.g., local sensitivities of local DP in \cite{lin2022ppbr} and individual keys of multi-key HE \cite{Robust2024Hao,ma2021pocket}).
The latter approach provides more reliable privacy protection but may suffer from high performance degradation, including increased communication/computation burdens and lossy model accuracy, when dealing with multi-parameter/key masked gradients.
Besides, the single server is susceptible to single points of failure.

\textbf{FL with the multi-server model.}
FL with the multi-server model, particularly the two-server model, gains attention in recent research on PBFL. In this model, users distribute their protected gradients across multiple non-colluding servers for secure aggregation. This approach mitigates single points of failure and, compared to single-server and server-less models, reduces the computational and communication burden on users without compromising model performance. While protocols with more than two servers can offer better efficiency, ensuring the validity of the honest-majority assumption remains challenges. Conversely,
for dishonest-majority protocols, two-server counterparts bring more performance improvements than more servers (see, e.g., comparison in Table 6 of \cite{braun2022motion}).
Liu et al. \cite{liu2021privacy} propose PEFL scheme based on the two-server model.
It employs linear HE to protect the privacy of local model updates and designs the corresponding secure similarity computation protocols between the two servers to detect poisonous gradient ciphertexts.
However, these secure computation protocols inadvertently disclose data privacy to the servers \cite{schneider2023comments}.
Zhang et al. \cite{zhang2022lsfl} propose LSFL scheme, which utilizes additive secret sharing to design a secret-sharing-based, secure Byzantine-robust aggregation protocol between the two servers.
This protocol enables poison detection and secure aggregation.
However, this scheme also suffers from privacy leakages during the implementation of secure aggregation between the two servers \cite{wu2023security}.
Ma et al. \cite{ma2022shieldfl} propose \nameScheme, a two-server model employing two-trapdoor HE and cosine similarity for encrypting gradients and detecting model poisoning. In \nameScheme, the secret key of the encryption algorithm is divided into two secret shares and distributed to two non-colluding servers separately. These servers then collaborate to detect poisonous encrypted gradients and obtain an aggregation gradient using a masking technique.
However, this scheme still violates privacy.
The privacy breaches of these two-server-based PBFL schemes arise from the secure computation processes of similarity measurements between the two servers. 
Specifically, the computation involves securely computing the similarity between a local gradient and a baseline gradient (e.g., the median gradient of all gradients in \cite{zhang2022lsfl} or the poisonous gradient baseline in \nameScheme), with the baseline gradient being identical across multiple computations. This identical baseline gradient makes it easy to design insecure schemes like \nameScheme and LSFL, which compromise privacy.

In this paper, we focus on PBFL schemes utilizing the two-server setting. These schemes fortify resilience against single points of failure and have the potential to improve FL performance regarding computational cost, communication burden, and final model accuracy. While Ma et al. \cite{ma2022shieldfl} and Zhang et al. \cite{zhang2022lsfl} have designed their PBFL schemes based on the two-server setting, their approaches unfortunately fall short in preserving privacy. Consequently, our scheme aims to rectify these privacy breaches and further improve computational and communication efficiency.

%% file: sections/2-problem.tex
\section{Problem Formulation}\label{sec:ProblemFormulation}

%\vspace*{-0.9\baselineskip}
\subsection{System Model}

Fig.~\ref{fig:system} illustrates the system model of our FL framework, comprising three types of entities: 
a key center \textit{KC},
$n$ users $U_{i\in[1,n]}$, 
and two servers $S_1$ and $S_2$.
\begin{itemize}[leftmargin=0.3cm]
  \item \textit{Key Center:} \textit{KC} generates and distributes public key and secret key shares to users and servers.

  \item \textit{Users:} In the $t$th training iteration, each user $U_{i}$ trains the local model $W^{(t)}$ with the individual dataset, encrypts its local gradient $g_i^{(t)}$ with the public key, and uploads the encrypted gradient $[\![g_i^{(t)}]\!]$ to $S_1$. Then, each user downloads the global gradient $g^{(t)}$ for local update.

  \item \textit{Servers:} $S_1$ and $S_2$ interact to detect poisoning of the users' local gradients $\{[\![g_i^{(t)}]\!]\}_{i\in[1,n]}$ and perform aggregating of the local gradients. %unpoisoned gradients.
\end{itemize}

\begin{figure}[t]
  \centering
  \includegraphics[width=8cm]{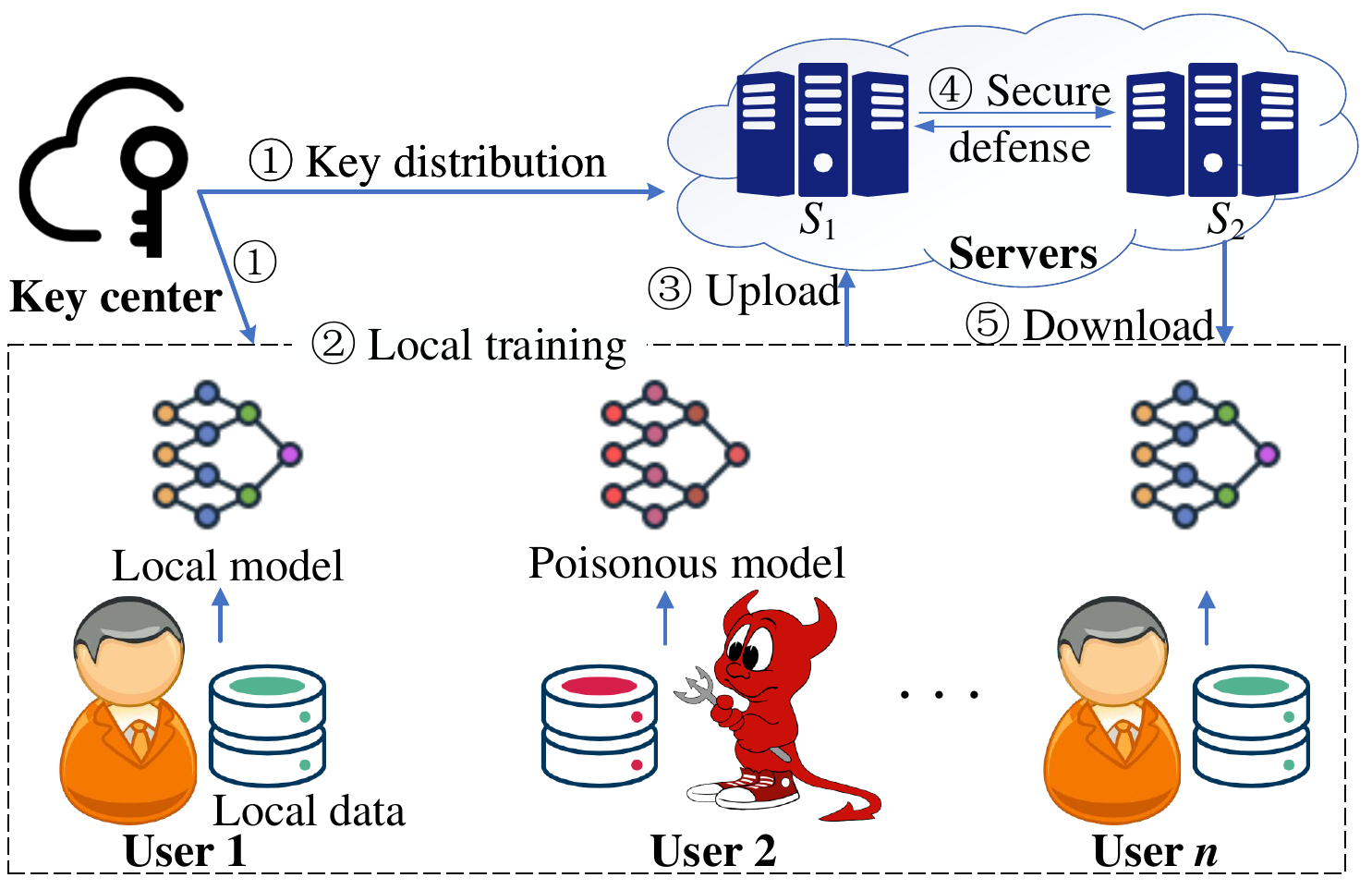}
  \caption{The system model}\label{fig:system}
\end{figure}

\subsection{Threat Model}

We consider a federated learning system comprising a set of $n$ clients $\{U_i\}_{i\in[1,n]}$ and two non-colluding servers $S_1,S_2$. The servers collaboratively manage encrypted aggregation and filtering tasks, while the clients periodically upload their local updates. We assume the presence of \textit{malicious clients} and at most \textit{one corrupted server}, a widely accepted assumption in dual-server frameworks~\cite{rathee2023elsa,gehlhar2023safefl}.
To capture different adversarial behaviors, we distinguish between \textit{honest-but-curious} and \textit{malicious} corruption types:
\begin{itemize}[leftmargin=0.3cm]
  \item \textbf{Honest-but-curious (semi-honest)} parties follow the protocol specification but attempt to infer private information from the data they receive.
  \item \textbf{Malicious} parties may arbitrarily deviate from the protocol, drop or fabricate messages, or tamper with the computation.
\end{itemize}

We analyze the following two adversarial scenarios:

\textbf{Case 1: Semi-honest server with malicious clients.}
In this setting, both servers behave semi-honest, while some clients may be malicious.
The semi-honest servers does not deviate from the protocol but attempts to learn sensitive information. Under this threat model, our scheme ensures:
\begin{itemize}[leftmargin=0.3cm]
  \item The \textit{privacy} of honest clients' local updates (e.g., gradients or model parameters) is preserved.
  \item The \textit{correctness} of server-side computations, including aggregation and filtering, is guaranteed.
  \item Malicious clients cannot compromise the privacy of others or inject invalid influence beyond their own inputs.
\end{itemize}

\textbf{Case 2: Malicious server with malicious clients.}
In this stronger setting, one of the servers may behave maliciously, while the other remains honest-but-curious. The corrupted server may alter, drop, or fabricate computation results. Nevertheless, as long as the two servers do not collude, our scheme provides:
\begin{itemize}[leftmargin=0.3cm]
  \item \textit{Privacy protection} of honest clients' updates remains intact due to cryptographic separation and dual control.
  \item However, the \textit{correctness of computation} cannot be unconditionally guaranteed, as the malicious server may tamper with intermediate results or deviate from the protocol.
\end{itemize}

This trade-off is common in dual-server designs~\cite{rathee2023elsa,gehlhar2023safefl}. Similar to SAFEFL~\cite{gehlhar2023safefl}, our framework is compatible with integrating \textit{verifiable computation techniques} (e.g., ZKP, MAC-based audits) to ensure correctness even in the presence of a malicious server. As verifiability introduces additional complexity and is orthogonal to our core design, we leave its implementation as future work.

\subsection{Design Goals}

The design goals are to ensure security and performance as detailed below.
\begin{itemize}[leftmargin=0.3cm]
\item \textbf{Confidentiality (Privacy Protection).} The scheme should be resistant to Type I attackers, ensuring that users' private datasets cannot be reconstructed by this type of attackers (including other users and servers).
\item \textbf{Robustness.} The scheme should be resistant to Type II attackers. It should be capable of detecting model poisoning attacks from malicious users based on users' encrypted local models and should be able to produce a high-quality global model.
\item \textbf{Efficiency.} The computational and communication overhead on users and servers should be minimized.
\end{itemize}

%% file: sections/2-vulnerabilities.tex
\section{Overview of Prior Work}\label{sec:oriScheme}
In this section, we briefly review the existing two-server-based PBFL scheme, \nameScheme~\cite{ma2022shieldfl}, and identify its security vulnerabilities.
\nameScheme employs a two-trapdoor partial homomorphic encryption (PHE) scheme based on Paillier's cryptosystem to preserve data privacy.
The specific construction of the two-trapdoor PHE is provided in Appendix~A.
\nameScheme consists of the following four processes: 

  1) \textit{Setup}:
      In this process, $KC$ generates a public/secret key pair $(pk,sk)$ and divides $sk$ into $n+1$ pairs of independent shares $(sk_1,sk_2)$, $\{(sk_{s_i},sk_{u_i})\}_{i\in[1,n]}$
      using $\mathbf{PHE.KeySplit}(sk)$.
      It then broadcasts \(pk\), distributes \(sk_1\) and \(sk_2\) to \(S_1\) and \(S_2\), respectively, and sends \(\{sk_{s_i}\}_{i\in[1,n]}\) and \(sk_{u_i}\) to \(S_1\) and each user \(U_i\), respectively.
      
  2) \textit{Local Training}: Each benign user $U_i$ trains its local model $W^{(t)}$, encrypts its local gradient $g_i^{(t)}$, and sends the encrypted gradient $[\![g_i^{(t)}]\!]\leftarrow\mathbf{PHE.Enc}_{pk}(g_i^{(t)})$ to $S_1$.
  A malicious user $U_i\in U^*$ may launch model poisoning and send encrypted poisonous gradients $[\![g_{i*}^{(t)}]\!]$ to $S_1$.

  3) \textit{Privacy-Preserving Defense}:
  This process mainly consists of three implementations.
  a) Normalization judgment. $S_1$ determines whether it received encrypted local gradients $[\![g_i^{(t)}]\!]$ are normalized.
  b) Poisonous Gradient Baseline Finding. $S_1$ finds the gradient $[\![g_*^{(t)}]\!]$ (among the users' local gradients) with the lowest cosine similarity as the baseline of poisonous gradients through implementing a secure cosine similarity measurement $\cos_{i\Sigma}\leftarrow\mathbf{SecCos}([\![g_i^{(t)}]\!],[\![g_{\Sigma}^{(t-1)}]\!])$ in Fig.~\ref{fig:SecCos}, where $[\![g_{\Sigma}^{(t-1)}]\!]$ is the aggregated gradient of the last iteration.
  c) Byzantine-Tolerance Aggregation.
     First, $S_1$ measures the cosine similarity between the poisonous baseline and each user's local gradient as $\cos_{i*}\leftarrow\mathbf{SecCos}([\![g_i^{(t)}]\!],[\![g_*^{(t)}]\!]).$
     Then, $S_1$ sets the confidence $co_i\leftarrow 1\times deg-\cos_{i*}$ for $g_i^{(t)}$
     and calculates $co\leftarrow\sum_{i=1}^nco_i, co_i=\lfloor\dfrac{co_i}{co}\cdot deg\rceil$. 
     Finally, $S_1$ generates a global gradient ciphertext $[\![g_{\Sigma}^{(t)}]\!]$,
     where $g_{\Sigma}^{(t)}=\sum_{i=1}^n co_i g_i^{(t)}$.

     % \item
      4) \textit{Model Update}.
      $S_1$ and $U_i$ interact to perform $\mathbf{PHE.PartDec}$ and $\mathbf{PHE.FullDec}$, so as $U_i$  obtains the global gradient $g_{\Sigma}^{(t)}$ and
      updates its local model as $W^{(t+1)}\leftarrow W^{(t)}-\eta^{(t)}g_{\Sigma}^{(t)}$, where $\eta^{(t)}$ is the learning rate of the $t$th training iteration.

\begin{fig}[t]
\small
Given two normalized ciphertexts $[\![g_a]\!]=\{[\![x_{a_1}]\!],\cdots,[\![x_{a_l}]\!]\}$ and $[\![g_b]\!]=\{[\![x_{b_1}]\!],\cdots,[\![x_{b_l}]\!]\}$, $S_1$ interacts with $S_2$ to obtain the cosine similarity $\cos_{ab}$ of $g_a$ and $g_b$: %through interacting with $S_2$.

\begin{itemize}[leftmargin=0.4cm]
  \item[$\bullet$] $@S_1$: On receiving $[\![g_a]\!]$ and $[\![g_b]\!]$, $S_1$ selects $l$ random noises $r_{i\in[1,l]}\leftarrow \mathbb{Z}_N^*$ and masks $[\![g_a]\!]$ and $[\![g_b]\!]$ as
      $[\![\bar{g}_a]\!]=\{[\![\bar{x}_{a_1}]\!],\cdots,[\![\bar{x}_{a_l}]\!]\}$ and $[\![g_b]\!]=\{[\![\bar{x}_{b_1}]\!],\cdots,[\![\bar{x}_{b_l}]\!]\}$, where $[\![\bar{x}_{a_i}]\!]=[\![x_{a_i}]\!]\cdot [\![r_i]\!]$ and $[\![\bar{x}_{b_i}]\!]=[\![x_{b_i}]\!]\cdot [\![r_i]\!]$.
      Then, $S_1$ partially decrypts the masked ciphertexts as
      $[\bar{x}_{a_i}]_1\leftarrow \mathbf{PHE.PartDec}_{sk_1}([\![\bar{x}_{a_i}]\!])$ and $[\bar{x}_{b_i}]_1\leftarrow \mathbf{PHE.PartDec}_{sk_1}([\![\bar{x}_{b_i}]\!]), s.t., [\![\bar{x}_{a_i}]\!]\in[\![\bar{g}_a]\!],[\![\bar{x}_{b_i}]\!]\in[\![\bar{g}_b]\!]$,
      and sends $[\bar{g}_a]_1,[\bar{g}_b]_1,[\![\bar{g}_a]\!],[\![\bar{g}_b]\!]$ to $S_2$.

\item[$\bullet$] $@S_2$: Once obtaining the above ciphertexts, $S_2$ calls $\mathbf{PHE.PartDec}$ and $\mathbf{PHE.FullDec}$ to perform full decryption and obtain the masked gradients as  $\bar{g}_a$ and $\bar{g}_b$. Then $S_2$ calculates the cosine similarity between $\bar{g}_a$ and $\bar{g}_b$ as
    $\overline{\cos}_{ab}=\bar{g}_a\odot\bar{g}_b=\sum_{i=1}^m\bar{x}_{a_i}\cdot\bar{x}_{b_i}$, where $\odot$ represents inner product.
    $S_2$ calls the encryption algorithm $\mathbf{PHE.Enc}$ to return $[\![\overline{\cos}_{ab}]\!]$ to $S_1$.

\item[$\bullet$] $@S_1$: $S_1$ removes the noise $r_i$ in $[\![\overline{\cos}_{ab}]\!]$ to obtain the final cosine similarity $[\![\cos_{ab}]\!]$ and sends it to $S_2$.

\item[$\bullet$] $@S_2$: $S_2$ computes the decryption share $[\cos_{ab}]_2\leftarrow\mathbf{PHE.PartDec}_{sk_2}([\![\cos_{ab}]\!])$ and returns it to $S_1$.

\item[$\bullet$] $@S_1$: On receiving $[\cos_{ab}]_2$, $S_1$ obtains the cosine similarity plaintext $\cos_{ab}$ by calling $\mathbf{PHE.PartDec}$ and $\mathbf{PHE.FullDec}$.
\end{itemize}
 \setcounter{fig}{\value{figure}}
 \caption{Implementation procedure of $\mathbf{SecCos}([\![g_a]\!],[\![g_b]\!])$.}
 \label{fig:SecCos}
\end{fig}

%% file: sections/3-analysis.tex
%\vspace*{-1\baselineskip}
%\section{Privacy Leakage of \nameScheme}\label{sec:attacks}
%%\vspace*{-1\baselineskip}
%In this section, we begin by outlining the privacy breaches observed in \nameScheme's two implementation procedures: poisonous gradient baseline finding and Byzantine-tolerance aggregation.
%Then, we illustrate the inadequacy of \nameScheme's provided security proof.
%Finally, upon analyzing the underlying reasons for these privacy violations, we propose an immediate fix for \nameScheme and demonstrate its ineffectiveness.

%%In this section, we first separately present the privacy breaches in \nameScheme's two implementation procedures, i.e., poisonous gradient baseline finding and Byzantine-tolerance aggregation.
%%Then, by combining the two procedures, we show that user gradients can even be leaked in clear.
%%Finally, we demonstrate how the security proof provided by \nameScheme is incorrect.
%In this section, we first present the privacy breaches in \nameScheme's two implementation procedures: poisonous gradient baseline finding and Byzantine-tolerance aggregation.
%Then, we demonstrate how the security proof provided by \nameScheme is incorrect.
%%
%Finally, after analyzing the reason for these privacy violations, we give an immediate fix for \nameScheme and present that this fix is ineffective.
%
%%The privacy-preserving defense strategy of \nameScheme violates privacy in two implementation phases: baseline poisonous gradient finding and Byzantine-tolerance aggregation.

%\vspace*{-0.5\baselineskip}
\subsection{Privacy Leakage in \nameScheme}
In \nameScheme, privacy leakages primarily occur during the privacy-preserving defense process. %We now describe this process in detail.
%In \nameScheme, the primary sources of privacy leakage arise during the privacy-preserving defense phase. We now describe this process in detail.
%
%\subsubsection{Poisonous Gradient Baseline Finding}
In the procedure of poisonous gradient baseline finding, $S_1$ interacts with $S_2$ to measure the secure cosine similarity $\cos_{i\Sigma}\leftarrow\mathbf{SecCos}([\![g_i^{(t)}]\!] ,[\![g_{\Sigma}^{(t-1)}]\!] )$ for $i\in[1,n]$, so as to find the baseline of poisonous gradient $[\![g_*^{(t)}]\!]\in\{[\![g_i^{(t)}]\!]\}_{i\in[1,n]}$ with the lowest cosine similarity.
Specifically, in this procedure, $S_1$ selects $l$ random noises $\{r_{i1},\cdots,r_{il}\}$ and
masks $[\![g_i^{(t)}]\!]=\{[\![x_{i1}]\!],\cdots,[\![x_{il}]\!]\}$ and $[\![g_{\Sigma}^{(t-1)}]\!]=\{[\![y_{1}]\!],\cdots,[\![y_{l}]\!]\}$ as $[\![\bar{g}_i^{(t)}]\!]$ and $[\![\bar{g}_{\Sigma,i}^{(t)}]\!]$ using Eq.~\eqref{eq:maskG} and Eq.~\eqref{eq:maskLastG}, respectively.
Then $S_1$ partially decrypts $[\![\bar{g}_i^{(t)}]\!]$ and $[\![\bar{g}_{\Sigma,i}^{(t)}]\!]$ as $[\bar{g}_i^{(t)}]_1$ and $[\bar{g}_{\Sigma,i}^{(t)}]_1$ and sends $[\![\bar{g}_i^{(t)}]\!],[\![\bar{g}_{\Sigma,i}^{(t)}]\!],[\bar{g}_i^{(t)}]_1,[\bar{g}_{\Sigma,i}^{(t)}]_1$ to $S_2$.

\vspace{-5pt}
{\small
\begin{align}
&\left(
 \begin{array}{c}
    [\![\bar{g}_1^{(t)}]\!] \\
    \vdots \\
    {[\![\bar{g}_i^{(t)}]\!]} \\
    \vdots \\
    {[\![\bar{g}_n^{(t)}]\!]} \\
 \end{array}
\right)  = 
\left(
 \begin{array}{ccc}
   [\![x_{11}]\!]\cdot [\![r_{11}]\!] & \cdots & [\![x_{1l}]\!]\cdot [\![r_{1l}]\!] \\
   \vdots & \cdot &\vdots \\
   {[\![x_{i1}]\!]\cdot [\![r_{i1}]\!]} & \cdots & [\![x_{il}]\!]\cdot [\![r_{il}]\!]\\
   \vdots & \cdot &\vdots \\
   {[\![x_{n1}]\!]\cdot [\![r_{n1}]\!]} & \cdots & [\![x_{nl}]\!]\cdot [\![r_{nl}]\!]\\
 \end{array}
\right)\label{eq:maskG}\\
&\left(
 \begin{array}{c}
   [\![\bar{g}_{\Sigma,1}^{(t-1)}]\!] \\
   \vdots \\
   {[\![\bar{g}_{\Sigma,i}^{(t-1)}]\!]} \\
   \vdots \\
   {[\![\bar{g}_{\Sigma,n}^{(t-1)}]\!]} \\
 \end{array}
\right) = 
\left(
 \begin{array}{ccccc}
   [\![y_1]\!]\cdot [\![r_{11}]\!] & \cdots & [\![y_l]\!]\cdot [\![r_{1l}]\!] \\
   \vdots & \cdot &\vdots \\
   {[\![y_1]\!]\cdot [\![r_{i1}]\!]} & \cdots & [\![y_l]\!]\cdot [\![r_{il}]\!]\\
   \vdots & \cdot &\vdots \\
   {[\![y_1]\!]\cdot [\![r_{n1}]\!]} & \cdots & [\![y_l]\!]\cdot [\![r_{nl}]\!]\\
 \end{array}
\right)\label{eq:maskLastG}
\end{align}}
\vspace{-5pt}

\textbf{Privacy Leakage:} 
%Upon receiving the ciphertexts from $S_1$, $S_2$ decrypts to obtain $\bar{g}_{i}^{(t)}$ and $\bar{g}_{\Sigma,i}^{(t-1)}$ for $i\in[1,n]$ as follows:
$S_2$ decrypts the ciphertexts received from $S_1$ to obtain $\bar{g}_{i}^{(t)}$ and $\bar{g}_{\Sigma,i}^{(t-1)}$ $(i\in[1,n])$ as follows:

%Upon receiving the ciphertexts $[\![\bar{g}_i^{(t)}]\!],[\![\bar{g}_{\Sigma,i}^{(t)}]\!],[\bar{g}_i^{(t)}]_1,[\bar{g}_{\Sigma,i}^{(t)}]_1$ from $S_1$, the server $S_2$ first invokes $\mathbf{PHE.PartDec}$ to obtain the partially decrypted value $\bar{g}_i^{(t)}$, and then it calls $\mathbf{PHE.FullDec}$ to reconstruct $\bar{g}_{\Sigma,i}^{(t-1)}$, following the procedures defined in Eq.~(\ref{eq:FullDecG} - \ref{eq:FullDecLastG}).

\vspace{-5pt}
{\small
\begin{align}
&\left(
 \begin{array}{c}
    \bar{g}_1^{(t)} \\
    \vdots \\
    {\bar{g}_i^{(t)}} \\
    \vdots \\
    {\bar{g}_n^{(t)}} \\
 \end{array}
\right)  =
\left(
 \begin{array}{ccc}
   x_{11}+ r_{11} & \cdots & x_{1l}+r_{1l} \\
   \vdots &  \cdot &\vdots \\
   {x_{i1}+r_{i1}} & \cdots & x_{il}+r_{il}\\
   \vdots & \cdot &\vdots \\
   {x_{n1}+r_{n1}} & \cdots & x_{nl}+r_{nl}\\
 \end{array}
\right)\label{eq:FullDecG}\\
&\left(
 \begin{array}{c}
   \bar{g}_{\Sigma,1}^{(t-1)} \\
   \vdots \\
   {\bar{g}_{\Sigma,i}^{(t-1)}} \\
   \vdots \\
   {\bar{g}_{\Sigma,n}^{(t-1)}} \\
 \end{array}
\right) =
\left(
 \begin{array}{ccc}
   y_1+r_{11} & \cdots & y_l+r_{1l} \\
   \vdots & \cdot &\vdots \\
   {y_1+r_{i1}} & \cdots & y_l+r_{il}\\
   \vdots & \cdot &\vdots \\
   {y_1+r_{n1}} & \cdots & y_l+r_{nl}\\
 \end{array}
\right)\label{eq:FullDecLastG}
\end{align}
\vspace{-5pt}

\noindent Then, $S_2$ can learn the difference between two arbitrary local gradients and the difference between the last-iteration aggregated gradient and an arbitrary local gradient: %by calculated Eq.~\eqref{eq:dij} and Eq.~\eqref{eq:diSigma}, respectively.
%The difference between the local gradients of the $i$th and $j$th users is calculated as
%The difference between the local gradient of the $i$th user and the local gradient of the $j$th user is calculated as
\begin{equation}\label{eq:dij}
  \begin{split}
    d_{ij}
    &=g_i^{(t)}-g_j^{(t)}
    =\{x_{i1}-x_{j1},\cdots,x_{il}-x_{jl}\}\\
    &=\left(\bar{g}_i^{(t)}-\bar{g}_j^{(t)}\right)-\left(\bar{g}_{\Sigma,i}^{(t-1)}-\bar{g}_{\Sigma,j}^{(t-1)}\right), i\neq j.\\
  \end{split}
\end{equation}
%
%Besides, $S_2$ can learn the difference between %the last-iteration aggregated gradient
%$g_{\Sigma}^{(t-1)}$ and an arbitrary gradient $g_i^{(t)}$: %, which is
\begin{equation}\label{eq:diSigma}
  \begin{split}
    d_{i\Sigma}
    &=g_i^{(t)}-g_{\Sigma}^{(t-1)}
    =\{x_{i1}-y_{1},\cdots,x_{il}-y_{l}\}\\
    &=\bar{g}_i^{(t)}-\bar{g}_{\Sigma,i}^{(t-1)}.
  \end{split}
\end{equation}
Thus, according to Eqs.  \eqref{eq:dij} and \eqref{eq:diSigma}, $S_2$ knows
\begin{equation}\label{eq:differenceG}
\small
\left(
 \begin{array}{c}
    g_1^{(t)} \\
    \vdots \\
    {g_i^{(t)}} \\
    \vdots \\
    {g_n^{(t)}} \\
 \end{array}
\right)  =
\left(
 \begin{array}{c}
   d_{1j}+g_j^{(t)}\\
   \vdots\\
   {d_{ij}+g_j^{(t)}}\\
   \vdots \\
   {d_{nj}+g_j^{(t)}}\\
 \end{array}
\right)  =
\left(
 \begin{array}{c}
   d_{1\Sigma}+g_{\Sigma}^{(t-1)}\\
   \vdots\\
   {d_{i\Sigma}+g_{\Sigma}^{(t-1)}}\\
   \vdots \\
   {d_{n\Sigma}+g_{\Sigma}^{(t-1)}}\\
 \end{array}
\right),
\end{equation}
in which
$S_2$ acquires a ``shifted" distribution of all local gradients in plaintext (It is worth noting that $\{d_{ij}\}_{i\in[1,n],j\in[1,n]}$ and $\{d_{i\Sigma}\}_{i\in[1,n]}$ are all accessible to $S_2$). This compromises privacy significantly \cite{geiping2020inverting}, %\cite{geiping2020inverting,schneider2023comments},
as it divulges a considerable amount of information to $S_2$, thereby failing to fulfill the confidentiality requirement outlined in \cite{ma2022shieldfl}.

The privacy leakage issue in the Byzantine-tolerant aggregation phase is similar to that in the poisonous gradient baseline finding process. 
For brevity, we present its details in Appendix~B. 
In addition, Appendix~B further provides a gradient reconstruction analysis, showing how the leaked information from both phases can be used to recover the complete plaintext gradients.

%\vspace*{-0.5\baselineskip}
\subsection{Immediate Fix and Its Ineffectiveness}\label{ssec:immediateFix}

%The reasons for the privacy breaches in \nameScheme mainly lie in two folds.

%(1) %In the secure cosine similarity measurement method $\mathbf{SecCos}$ of \nameScheme,
%In $\mathbf{SecCos}$,
The main reason for the privacy breaches in \nameScheme is that 
the local gradients $g_i^{(t)}$ and the last-iteration aggregated gradient $g_{\Sigma}^{(t-1)}$
(or the poisonous gradient baseline $g_*^{(t)}$) are masked
%the local gradients $g_i^{(t)}=\{x_{i1},\cdots,x_{il}\}$ and the last-iteration aggregated gradient $g_{\Sigma}^{(t-1)}=\{y_1,\cdots,y_l\}$
%(or the poisonous gradient baseline $g_*^{(t)}=\{z_1,\cdots,z_l\}$) are masked
with the same random noise vector, 
thereby exposing the correlations between these gradients and resulting in the severe consequence that once one of these gradients is leaked, all other gradients can be directly revealed in clear.

Although different random noise vectors can be used %on the local gradients, the last-iteration aggregated gradient, and the poisonous gradient baseline 
to prevent the aforementioned privacy violations, using the same noise vector is inevitable to maintain the correctness of %the secure cosine similarity measurement 
$\mathbf{SecCos}$.
%
%While the breaches above could be prevented by using different random noises across the local gradients, the last-iteration aggregated gradient and the poisonous gradient baseline,
%we emphasize that the use of the same noises is unavoidable for $\mathbf{SecCos}([\![g_a]\!],[\![g_b]\!])$ to remain correct.
In detail, the secure cosine similarity measurement between two gradient ciphertexts $[\![g_a]\!]$ and $[\![g_b]\!]$ is calculated by first masking the two gradients with random noise vectors $r$ and $\gamma$ to output $[\![\bar{g}_a]\!]=[\![g_a+r]\!]$ and $[\![\bar{g}_b]\!]=[\![g_b+\gamma]\!]$, respectively, then decrypting the masked gradients and calculating the cosine similarity
$\overline{\cos}_{ab}=\frac{\bar{g}_a\odot\bar{g}_b}{|\!|\bar{g}_a|\!|\cdot |\!|\bar{g}_b|\!|}=\frac{(g_a+r)\odot(g_b+\gamma)}{|\!|g_a+r|\!|\cdot |\!|g_b+\gamma|\!|}$.
However, $\frac{\bar{g}_a\odot\bar{g}_b}{|\!|\bar{g}_a|\!|\cdot |\!|\bar{g}_b|\!|}= \frac{g_a\odot g_b}{|\!|g_a|\!|\cdot |\!|g_b|\!|}$ holds only when $r=\gamma$, where $\frac{g_a\odot g_b}{|\!|g_a|\!|\cdot |\!|g_b|\!|}$ is the correct computation of the cosine similarity between $g_a$ and $g_b$; Otherwise, if $r\neq\gamma$, the correctness will be breached.

To prevent privacy violations and ensure the correctness of computations, we deviate from \nameScheme of employing distinct random noise vectors to obscure gradient ciphertexts.
%in its implementation procedure of $\mathbf{SecCos}([\![g_a]\!],[\![g_b]\!])$.
Instead, we introduce a novel method where we mask the computation result of gradient ciphertexts with a random noise vector, as detailed in our new secure computation methods outlined in Section \ref{ssec:NSCM}. This approach leverages a newly constructed two-trapdoor Fully Homomorphic Encryption (FHE) algorithm as the fundamental cryptographic building block.

%we construct a two-trapdoor FHE algorithm.
%Building upon the two-trapdoor FHE, we propose new secure computation methods and design an enhanced scheme.

%Instead of using different random noise vectors, to prevent privacy violations and ensure the correctness of computations, we next construct a two-trapdoor FHE algorithm and based on which, we propose enhanced secure computation methods and design a secure and efficient PBFL scheme.
%design an enhanced scheme using newly proposed secure computation methods.

%To prevent the privacy violations and ensure the correctness of the computations, we next design an enhanced scheme using newly proposed secure computation methods.

%\section{Conclusion}
%
%In this paper, we analyze the security of Ma et al.'s privacy-preserving federated learning model \nameScheme.
%We show that \nameScheme violates privacy it claimed.
%The correlation among arbitrary local gradients can be obtained by an honest-but-curious server.
%Moreover, the server can even infer all users' gradients in clear, thereby breaching privacy of the entire model.
%%It exists universal forgery, which makes that anyone can forge a false traffic emergency message and its real identity cannot be traced.

%% file: sections/4-enhanced.tex
\section{The Enhanced Scheme}\label{sec:enhanced}

In this section, we propose an enhanced PBFL scheme based on the two-server setting and effectively addresses the privacy leakages identified in \nameScheme.
Additionally, the enhanced scheme demonstrates improved efficiency in both computational and communication performance.

\subsection{Technical Intuition}

\textbf{We empower our enhanced scheme to resist privacy attacks while ensuring computation correctness.}
In \nameScheme's secure computation methods, the server $S_1$ applies a masking technique to the received user gradient ciphertexts.
Then $S_1$ and $S_2$ interact to compute on the masked ciphertexts, and finally $S_2$ partially decrypts the computation results of the masked ciphertexts and sends them back to $S_1$. $S_1$ obtains the masked computation results and demasks to acquire the final computation plaintext output.
Although the masking technique is highly efficient, ensuring the correctness of computation requires that all ciphertexts be masked using the same random noise. This introduces strong correlations among masked values: if any plaintext is leaked, all correlated gradients can be inferred, thereby violating privacy.
A straightforward solution, using independent noise for each ciphertext, would eliminate these correlations but compromise computation correctness.

To simultaneously ensure both security and correctness, we fundamentally redesign the secure computation paradigm in our enhanced scheme.
Specifically, we introduce a novel masking protocol build upon a newly constructed two-trapdoor FHE scheme.
Instead of masking the inputs, we mask the computation results using different random vectors. This eliminates the need for shared noise and preserves correctness.
To handle ciphertext operations under masking, we classify them into linear (e.g., addition, subtraction, scalar multiplication) and non-linear (e.g., multiplication, exponentiation) categories. Linear operations on masked ciphertexts yield correctly demaskable results. In contrast, non-linear operations introduce complex noise that is difficult to eliminate.
To mitigate this without adding computational overhead, we perform non-linear operations directly on the original ciphertexts at $S_1$. The results are then masked and sent to $S_2$, which performs the remaining linear operations and partial decryption. Finally, $S_1$ removes the mask to obtain the correct plaintext output.
In this design, $S_1$ only observes the final computation results, while $S_2$ only accesses masked intermediate values. This separation effectively defends against the proposed privacy attacks.

Building on the technical intuition outlined above, we design two secure computation methods between the servers $S_1$ and $S_2$: an enhanced secure normalization judgment and an enhanced secure cosine similarity measurement, as introduced in Section~\ref{ssec:NSCM}. Notably, our design principle can also be extended to develop other secure computation methods, such as secure Euclidean distance, secure mean, and secure median, by separating the computation into linear and non-linear operations and employing our encryption and masking techniques within the two-server interaction framework.

\textbf{We improve both computational and communication efficiency in the enhanced scheme.}
This leap in efficiency stems from a couple of key insights:

(1) \textit{Ciphertext computation simplicity:} In our enhanced scheme, the secure computation process of non-linear operations on ciphertexts -- crucial for operations like computing ${\cos}_{ab} = g_a \odot g_b$ during $\mathbf{SecCos}([\![g_a]\!],[\![g_b]\!])$ -- is streamlined. Unlike previous methods that relied on Paillier's addition homomorphism and complex masking/unmasking techniques, requiring multiple interactions between the two servers, we construct a two-trapdoor FHE that allows for direct non-linear operations on ciphertexts.
This direct approach significantly reduces both the number of computation steps and the communication rounds, making our secure computation methods more efficient than those in \nameScheme.
For example, our secure computation methods (outlined in Figs.~\ref{fig:NewSecJudge} and \ref{fig:NewSecCos}) only require two rounds of interaction between $S_1$ and $S_2$, whereas the original scheme requires four rounds, reducing the communication burden by (at least) half.

(2) \textit{Cryptographic primitive efficiency:}
We employ an RLWE-based FHE, specifically the CKKS cryptosystem, to construct our two-trapdoor FHE scheme.
This choice is more efficient than the Paillier cryptosystem-based two-trapdoor PHE. The RLWE-based system, with its smaller ciphertext modulus requirements, offers lower computational overhead and communication traffic per step. For instance, at a $128$-bit security level, the two-trapdoor PHE demands a ciphertext modulus $N^2$ of at least $3072\times2$ bits, whereas the RLWE-based FHE can operate with a much smaller $q$ of $56$ bits (with $n_f=2048$), significantly reducing the computational cost.
We provide a detailed performance comparison in Section~\ref{ssec:performanceComp}.

(3) \textit{No additional trade-offs:}
Importantly, these improvements do not introduce extra limitations. Our enhanced scheme retains the same flexibility and scalability as \nameScheme because both adopt a two-trapdoor architecture and the same key distribution mechanism.
Although CKKS introduces negligible approximation error due to its encoding and encryption, such error is naturally tolerated by stochastic machine learning models and does not impact convergence or accuracy, as verified by our experimental results in Section~\ref{sec:experiment}.

In summary, our enhanced scheme introduces significant improvements in both efficiency and security without sacrificing flexibility, accuracy, or scalability.
By employing randomly selected noise for masking and categorizing operations into linear and non-linear types, we effectively address the privacy concerns associated with previous methods.
The streamlined process for non-linear operations, facilitated by a two-trapdoor FHE based on the RLWE homomorphic encryption scheme, minimizes the complexity and communication overhead typically required.
This innovative approach not only ensures the correctness of the computations but also bolsters resistance to privacy attacks, setting a new benchmark for secure computation methods.
Ultimately, our advancements position our scheme as a more efficient and secure alternative to existing solutions like \nameScheme.

\subsection{New Secure Computation Methods}\label{ssec:NSCM}

Our new secure computation methods consist of an enhanced secure normalization judgment and an enhanced secure cosine similarity measurement.
We construct a two-trapdoor FHE scheme to protect the privacy of these two methods.
The two-trapdoor FHE is based on RLWE (Ring Learning with Errors) homomorphic encryption (CKKS cryptosystem \cite{cheon2017homomorphic}), mainly containing the following algorithms.

\begin{fig}[!ht]
\small
Given a ciphertext $[\![g_{i}]\!]=\{[\![g_{i,1}]\!],\cdots,[\![g_{i,\tau}]\!]\}$ of a local gradient $g_i$,
where $\tau=\lceil l/n_f\rceil$ and $g_{i,j}=\{x_{i,j,1},\cdots,x_{i,j,n_f}\}$ for $j\in[1,\tau]$,
$S_1$ interacts with $S_2$ to judge whether the local gradient $g_i$ is normalized:
\begin{itemize}[leftmargin=0.4cm]
  \item[$\bullet$] $@S_1$: %On receiving $[\![g_i]\!]$,
  $S_1$ first calculates the component-wise second power of $[\![g_{i}]\!]$ as $[\![g_{i}^{2}]\!]=\{[\![g^2_{i,1}]\!],\cdots,[\![g^2_{i,\tau}]\!]\}$, where $[\![g^2_{i,j}]\!]\leftarrow\mathbf{FHE.Mult}_{evk}([\![g_{i,j}]\!],[\![g_{i,j}]\!])$ for $j\in[1,\tau]$.
  To protect gradient privacy, $S_1$ selects $\tau$ random noise polynomials $r_{i,j}=\{r_{i,j,1},\cdots,r_{i,j,n_f}\}\leftarrow\mathcal{R}_q$ for ${j\in[1,\tau]}$ and masks $[\![g_{i,j}^{2}]\!]$ with $r_{i,j}$ as $[\![\bar{g}_{i,j}^{2}]\!]=[\![g_{i,j}^{2}+r_{i,j}]\!]\leftarrow\mathbf{FHE.Add}([\![g_{i,j}^{2}]\!],[\![r_{i,j}]\!]).$
  Then, $S_1$ executes $[\bar{g}_{i,j}^{2}]_1\leftarrow \mathbf{FHE.PartDec}_{sk_1}([\![\bar{g}_{i,j}^{2}]\!])$ for $j\in[1,\tau]$ and sends
  $\{[\bar{g}_{i,j}^{2}]_1,[\![\bar{g}_{i,j}^{2}]\!]\}_{j\in[1,\tau]}$ to $S_2$.

\item[$\bullet$] $@S_2$: After obtaining the above ciphertexts, 
$S_2$  obtains the masked polynomials $\bar{g}_{i}^2=\{\bar{g}_{i,1}^2,\cdots,\bar{g}_{i,\tau}^2\}$ 
by executing $[\bar{g}_{i,j}^{2}]_2\leftarrow \mathbf{FHE.PartDec}_{sk_2}([\![\bar{g}_{i,j}^{2}]\!])$ and $\bar{g}_{i,j}^{2}\leftarrow\mathbf{FHE.FullDec}([\![\bar{g}_{i,j}^{2}]\!],[\bar{g}_{i,j}^{2}]_1,[\bar{g}_{i,j}^{2}]_2)$ for $j\in[1,\tau]$.
According to the multiplicative homomorphic properties of the two-trapdoor FHE, it is held that $\bar{g}_{i,j}^{2}=\{\bar{x}^2_{i,j,1},\cdots,\bar{x}^2_{i,j,n_f}\}$, where $\bar{x}^2_{i,j,k}=x^2_{i,j,k}+r_{j,k}$ for $k\in[1,n_f]$.
Then, $S_2$ returns $\overline{sum}=\sum_{j=1}^{\tau}\sum_{k=1}^{n_f}\bar{x}_{i,j,k}^2$ to $S_1$.

\item[$\bullet$] $@S_1$:  On receiving $\overline{sum}$,
$S_1$ removes the random noises in $\overline{sum}$  by executing $sum=\overline{sum}-\sum_{j=1}^{\tau}\sum_{k=1}^{n_f}r_{i,j,k}.$ 
Then $S_1$ determines whether $sum=1$.
If yes, it means that $g_i$ is normalized; otherwise not.

\end{itemize}
 \setcounter{figure}{\value{fig}}
 \caption{Implementation procedure of the enhanced secure normalization judgement method $\mathbf{ESecJudge}([\![g_{i}]\!])$.}
 \label{fig:NewSecJudge}
\end{fig}

\begin{itemize}[leftmargin=0.3cm]
  \item $\mathbf{FHE.Setup}(1^{\varepsilon'})\rightarrow PP$: Given the security parameter $\varepsilon'$, the public parameter $PP=\{n_f, p, q,P,\chi_s,\chi_e\}$ is output, where $n_f$ is a ring dimension, $p$ is a plaintext modulus, $q=q(\varepsilon')$ is a ciphertext modulus, $P$ is a special modulus, and $\chi_s,\chi_e$ are two probability distributions on the ring $\mathcal{R}$ with small bounds.

  \item %Key generation algorithm
  $\mathbf{FHE.KeyGen}(PP) \rightarrow (sk, pk,evk)$:
  Given the public parameter $PP$, a secret key $sk\leftarrow\chi_s$, a public key $pk\leftarrow(-a\cdot sk+e \bmod q, a)\in \mathcal{R}_q^2$, and an evaluation key $evk\leftarrow(-a'\cdot sk +e'+P\cdot sk^2 \bmod P\cdot q, a')\in \mathcal{R}_{P\cdot q}^2$ are output, where $a,a'\leftarrow \mathcal{R}_q:=\mathbb{Z}_q/(X^{n_f}+1)$ and $e,e'\leftarrow\chi_e$.

  \item $\mathbf{FHE.KeySplit}(sk) \rightarrow (sk_1,sk_2)$:
  Given the secret key $sk$, two secret key shares $sk_1,sk_2$ are generated by randomly dividing $sk$ meeting $sk=sk_1+sk_2\bmod q$.

  \item $\mathbf{FHE.Enc}_{pk}(m) \rightarrow [\![m]\!]$:
  This algorithm first encodes a float vector $m\in\mathbb{R}^{n_f}$ with length $n_f$ into an integral coefficient polynomials $m(X)\in\mathcal{R}_p$ and then it encrypts $m(X)$ and outputs the ciphertext $[\![m]\!]$.
      Specifically, the input vector $m$ is first converted to a complex vector $z=(z_1,\cdots,z_{n_f/2})\in\mathbb{C}^{n_f/2}$, in which every two elements in $m$ form the real and imaginary parts of a complex number.
      Then, the complex vector $z$ is encoded as $m(X)=\lfloor\Delta\cdot \phi^{-1}(z)\rceil$, where $\Delta>1$ is a scaling factor and $\phi: R[X]/(X^{n_f}+1)\rightarrow\mathbb{C}^{n_f/2}$ is a ring isomorphism.
      Finally, the ciphertext $[\![m]\!]=(m+u\cdot pk_0+e_0 \bmod q, u\cdot pk_1+e_1 \bmod q)\triangleq (c_0,c_1)\in\mathcal{R}_q^2$ is output, where $u\leftarrow\chi_s$ and $e_0,e_1\leftarrow\chi_e$.

  \item $\mathbf{FHE.Add}([\![m_1]\!],[\![m_2]\!]) \rightarrow [\![m]\!]$: Given two ciphertexts $[\![m_1]\!],[\![m_2]\!]$,
  the sum $[\![m]\!]=[\![m_1]\!]+[\![m_2]\!]=[\![m_1+m_2]\!]$ of the two ciphertexts is output.

  \item $\mathbf{FHE.Mult}_{evk}([\![m_1]\!],[\![m_2]\!]) \rightarrow [\![m]\!]$: Given two ciphertexts $[\![m_1]\!]=(c_0,c_1),[\![m_2]\!]=(c'_0,c'_1)$,
  let $(d_0,d_1,d_2)=(c_0c'_0,c_0c'_1+c'_0c_1,c_1c'_1) \bmod q$, the product $[\![m]\!]=(d_0,d_1)+\lfloor P^{-1}\cdot d_2\cdot evk\rceil$ of the two ciphertexts is output.

  \item %Partial decryption algorithm
  $\mathbf{FHE.PartDec}_{sk_i}([\![m]\!]) \rightarrow [m]_i$: Given a ciphertext $[\![m]\!]=(c_0,c_1)$,
  a partial decryption result $[m]_i=c_1sk_i + e_i \bmod q$ is output, where $e_i\leftarrow\chi_e$.

  \item $\mathbf{FHE.FullDec}([\![m]\!],[m]_1,[m]_2) \rightarrow m$: Given a ciphertext $[\![m]\!]=(c_0,c_1)$ and its partial decryption result pair $([m]_1,[m]_2)$, the algorithm first fully decrypts the ciphertext $[\![m]\!]$ as a plaintext polynomial $m(X)=c_0+[m]_1+[m]_2 \bmod p\in\mathcal{R}_p$, then it decodes $m(X)$ by first computing the complex vector $z=\Delta^{-1}\cdot\phi(m(X))\in\mathbb{C}^{n_f/2}$ and then converting $z$ to a float vector $m\in\mathbb{R}^{n_f}$.
      Finally, the algorithm outputs $m$.
\end{itemize}

\begin{fig}[t]
\small
Given two normalized ciphertexts $[\![g_{a}]\!]=\{[\![g_{a,1}]\!],\cdots,[\![g_{a,\tau}]\!]\}$ and $[\![g_{b}]\!]=\{[\![g_{b,1}]\!],\cdots,[\![g_{b,\tau}]\!]\}$ of the gradients $g_a$ and $g_b$, respectively, where $g_{a,j}=\{x_{a,j,1},\cdots,x_{a,j,{n_f}}\}$ and $g_{b,j}=\{x_{b,j,1},\cdots,x_{b,j,{n_f}}\}$ for $j\in[1,\tau]$,
$S_1$ interacts with $S_2$ to obtain the cosine similarity $\cos_{ab}$ of $g_a$ and $g_b$: %through interacting with $S_2$.
\begin{itemize}[leftmargin=0.4cm]
  \item[$\bullet$] $@S_1$: On receiving $[\![g_a]\!]$ and $[\![g_b]\!]$, $S_1$ calculates the product of the two ciphertexts as
  $[\![g_{c,j}]\!]=[\![g_{a,j}]\!]\times[\![g_{b,j}]\!]\leftarrow\mathbf{FHE.Mult}_{evk}([\![g_{a,j}]\!], [\![g_{b,j}]\!])$ for $j\in[1,\tau]$.
  To protect gradient privacy, $S_1$ selects $\tau$ random noise polynomials $r_j=\{r_{j,1},\cdots,r_{j,n_f}\}\leftarrow\mathcal{R}_q$ for ${j\in[1,\tau]}$ and masks $[\![g_{c,j}]\!]$ with $r_j$ as
  $[\![\bar{g}_{c,j}]\!]=[\![g_{c,j}+r_j]\!]\leftarrow\mathbf{FHE.Add}([\![g_{c,j}]\!],[\![r_j]\!]).$
  Then, $S_1$ calculates $[\![\bar{g}_{c,\Sigma}]\!]=[\![\sum_{j=1}^{\tau}\bar{g}_{c,j}]\!]\leftarrow\mathbf{FHE.Add}([\![\bar{g}_{c,1}]\!], \cdots,[\![\bar{g}_{c,\tau}]\!])$ and partially decrypts it as $[\bar{g}_{c,\Sigma}]_1\leftarrow \mathbf{FHE.PartDec}_{sk_1}([\![\bar{g}_{c,\Sigma}]\!])$.
  $S_1$ sends $[\bar{g}_{c,\Sigma}]_1,[\![\bar{g}_{c,\Sigma}]\!]$ to $S_2$.

\item[$\bullet$] $@S_2$:
After receiving $[\bar{g}_{c,\Sigma}]_1,[\![\bar{g}_{c,\Sigma}]\!]$,
$S_2$ obtains $\bar{g}_{c,\Sigma}=\{x_{c,j,1},\cdots,x_{c,j,n_f}\}$ by calling $\mathbf{FHE.PartDec}$ and $\mathbf{FHE.FullDec}$.
According to the additive and multiplicative homomorphic properties of the two-trapdoor FHE, it is held that
$x_{c,j,k}=\sum_{j=1}^{\tau}(x_{a,j,k}x_{b,j,k}+r_{j,k})$.
Then, $S_2$ calculates the masked cosine similarity $\overline{\cos}_{ab}=\sum_{k=1}^{n_f}\sum_{j=1}^{\tau}x_{c,j,k}=\sum_{k=1}^{n_f}\sum_{j=1}^{\tau}(x_{a,j,k}x_{b,j,k}+r_{j,k})$ and sends it to $S_1$.

\item[$\bullet$] $@S_1$:
On receiving $\overline{\cos}_{ab}$, $S_1$ removes the random noises in $\overline{\cos}_{ab}$ by executing ${\cos}_{ab}=\overline{\cos}_{ab}-\sum_{k=1}^{n_f}\sum_{j=1}^{\tau}r_{j,k}$.

\end{itemize}
 \caption{Implementation procedure of the enhanced secure cosine similarity measurement method $\mathbf{ESecCos}([\![g_{a}]\!],[\![g_{b}]\!])$.}
 \label{fig:NewSecCos}
\end{fig}

Based on the two-trapdoor FHE scheme, we design the enhanced secure normalization judgment method $\mathbf{ESecJudge}$ and the enhanced secure cosine similarity measurement method $\mathbf{ESecCos}$ as Fig.~\ref{fig:NewSecJudge} and Fig.~\ref{fig:NewSecCos}, respectively.

For $\mathbf{ESecJudge}$, on inputting a gradient ciphertext $[\![g_i]\!]=\{[\![g_{i,1}]\!],\cdots,[\![g_{i,\tau}]\!]\}$, $S_1$ interacts with $S_2$ to judge whether the local gradient $g_i$ is normalized, i.e., whether $[\![g_{i,1}^2]\!]+\cdots+[\![g_{i,\tau}^2]\!]=[\![1]\!]$.
In this ciphertext computation, the linear and non-linear operations correspond to additions and multiplications, respectively.
We leverage $S_1$ to perform multiplication operations locally. The resulting computations are then masked and sent to $S_2$ for the addition operations.

Specifically, $S_1$ first computes ciphertext multiplications $[\![g^2_{i,j}]\!]\leftarrow\mathbf{FHE.Mult}_{evk}([\![g_{i,j}]\!],[\![g_{i,j}]\!])$ for $j\in[1,\tau]$ to obtain
$[\![g_{i}^{2}]\!]=\{[\![g^2_{i,1}]\!],\cdots,[\![g^2_{i,\tau}]\!]\}$. It then masks the result $[\![g_{i}^{2}]\!]$ as
 $[\![\bar{g}_{i}^{2}]\!]=\{[\![\bar{g}^2_{i,1}]\!],\cdots,[\![\bar{g}^2_{i,\tau}]\!]\}$, where
$[\![\bar{g}_{i,j}^{2}]\!]=[\![g_{i,j}^{2}+r_{i,j}]\!]\leftarrow\mathbf{FHE.Add}([\![g_{i,j}^{2}]\!],[\![r_{i,j}]\!])$ and $r_{i,j}$ is a randomly selected noise polynomial.
$S_1$ partially decrypts $\{[\![\bar{g}_{i,j}^{2}]\!]\}_{j\in[1,\tau]}$ and sends $\{[\bar{g}_{i,j}^{2}]_1,[\![\bar{g}_{i,j}^{2}]\!]\}_{j\in[1,\tau]}$ to $S_2$.
After performing full decryption, $S_2$ obtains the masked gradient $\bar{g}_{i}^2=\{\bar{g}_{i,1}^2,\cdots,\bar{g}_{i,\tau}^2\}$ and adds all components of the masked gradient (polynomial) to obtain $\overline{sum}=\sum_{j=1}^{\tau}\sum_{k=1}^{n_f}\bar{x}_{i,j,k}^2$, where $\bar{g}_{i,j}^{2}=\{\bar{x}^2_{i,j,1},\cdots,\bar{x}^2_{i,j,n_f}\}$ and
$\bar{x}^2_{i,j,k}$ is the $k$th polynomial coefficient of $\bar{g}_{i,j}^{2}$.
$S_2$ sends $\overline{sum}$ to $S_1$, who then obtains $sum=\overline{sum}-\sum_{j=1}^{\tau}\sum_{k=1}^{n_f}r_{i,j,k}$ through removing noises.
Finally, if $sum=1$, $S_1$ determines that $g_i$ is normalized and accepts it; otherwise, $S_1$ discards the unnormalized $g_i$.
Fig.~\ref{fig:NewSecJudge} illustrates the detailed procedure of $\mathbf{ESecJudge}$.

For $\mathbf{ESecCos}$, on inputting two normalized gradient ciphertexts $[\![g_{a}]\!]=\{[\![g_{a,1}]\!],\cdots,[\![g_{a,\tau}]\!]\}$ and $[\![g_{b}]\!]=\{[\![g_{b,1}]\!],\cdots,[\![g_{b,\tau}]\!]\}$ with $\parallel g_a\parallel=1$ and $\parallel g_b\parallel=1$, $S_1$ interacts with $S_2$ to calculate the cosine similarity of $g_a$ and $g_b$ as $\cos_{ab}=\frac{g_a\odot g_b}{\parallel g_a\parallel\times\parallel g_b\parallel}=g_{a,1}\times g_{b,1}+\cdots+g_{a,\tau}\times g_{b,\tau}$.
In this ciphertext computation, the linear and non-linear operations correspond to additions and multiplications, respectively.
We leverage $S_1$ to perform multiplication operations locally. The resulting computations are then masked and sent to $S_2$ for the addition operations.

In particular, first, $S_1$ performs the ciphertext multiplication and masking to obtain $[\![\bar{g}_{c,\Sigma}]\!]=\sum_{j=1}^{\tau}[\![g_{a,j} \times g_{b,j}+r_j]\!]$, where $\{r_j\}_{j\in[1,\tau]}$ are randomly selected noise polynomials by $S_1$.
Then, $S_2$ obtains the plaintext $\bar{g}_{c,\Sigma}=\{x_{c,j,1},\cdots,x_{c,j,n_f}\}$ by jointly decrypting $[\![\bar{g}_{c,\Sigma}]\!]$ with $S_1$ and sends the masked cosine similarity $\overline{\cos}_{ab}=\sum_{k=1}^{n_f}\sum_{j=1}^{\tau}x_{c,j,k}$ to $S_1$, where $x_{c,j,k}=\sum_{j=1}^{\tau} (g_{a,j,k} \times g_{b,j,k}+r_{j,k})$.
Finally, $S_1$ obtains $\cos_{ab}$ through removing noises.
A more specified $\mathbf{ESecCos}$ procedure is described in Fig.~\ref{fig:NewSecCos}.

%\vspace*{-1\baselineskip}
\subsection{Construction of Our PBFL Scheme}
In our enhanced scheme, a privacy-preserving defense method are devised to fortify against model poisoning attacks, in which
the new secure computation methods are utilized to protect the privacy while ensuring the correctness of the computations.
Besides, to motivate users to engage in benign training and accurately identify poisonous gradients, inspired by \cite{zhang2022fldetector}, we introduce user credit scores to our Byzantine-tolerance aggregation process of the privacy-preserving defense.
Specifically, the enhanced scheme contains four processes detailed below.

%\begin{enumerate}[leftmargin=0.4cm]
%  \item
  1) \textit{Setup}: %This process initializes the system by distributing keys and an initial model.
  In this process, $KC$ generates public parameters/keys for the scheme and $S_1$ generates an initial model for the users.
  Specifically, given a security parameter $\varepsilon'$, $KC$ calls $\mathbf{FHE.Setup}$ and $\mathbf{FHE.KeyGen}$ to generate the public parameter $PP$ and a tuple of key $(sk,pk,evk)$.
  Then, $KC$ calls $\mathbf{FHE.KeySplit}(sk)$ $n+1$ times to generate $n+1$ pairs of secret key shares $(sk_1,sk_2)$ and $\{(sk_{s_i},sk_{u_i})\}_{i\in[1,n]}$.
  Finally, $KC$ broadcasts $\{PP,pk,evk\}$, distributes $sk_1,\{sk_{s_i}\}_{i\in[1,n]}$ to $S_1$, distributes $sk_2$ to $S_2$, and distributes $sk_{u_i}$ to $U_i$.
  $S_1$ selects a random initial model $W^{(0)}$ and sends $W^{(0)}$ to $U_i$.

 % \item
  2) \textit{Local Training}:
  In each training iteration, each user trains its local model to obtain a local gradient, encrypts the local gradient, and sends the encrypted gradient to $S_1$.
  Specifically, in the $t$th training iteration, each user $U_i$ trains its local model $W^{(t)}$ and flats its local gradient as a gradient vector $g_i^{(t)}\in\mathbb{R}^l$ whose length is $l$.
  The gradient vector is then normalized and uniformly divided into $\tau=\lceil l/n_f\rceil$ vectors $\{g_{i,j}^{(t)}\in\mathbb{R}^{n_f}\}_{j\in[1,\tau]}$ with length $n_f$.
  The elements 0s are padded at the end of $g_i^{(t)}$ when less than $n_f$ elements are left.
  Then, $U_i$ encrypts each vector $g_{i,j}^{(t)}$ as $[\![g_{i,j}^{(t)}]\!]$ by calling $\mathbf{FHE.Enc}$.
  Finally, the beginning user $U_i$ sends its encrypted gradient $[\![g_i^{(t)}]\!]=\{[\![g_{i,1}^{(t)}]\!],\cdots,[\![g_{i,\tau}^{(t)}]\!]\}$ to $S_1$ and the malicious user $U_i\in U^*$ may launch poisoning attacks and then send its encrypted poisonous gradients $[\![g_{i*}^{(t)}]\!]=\{[\![g_{i*,1}^{(t)}]\!],\cdots,[\![g_{i*,\tau}^{(t)}]\!]\}$ to $S_1$.

\textit{3) Privacy-Preserving Defense:}
This process mainly contains the following three implementations.

\begin{itemize}[leftmargin=1em]
  \item \textbf{Normalization Judgment.}
  The server $S_1$ determines whether the received encrypted local gradients $[\![g_i^{(t)}]\!] = \{[\![g_{i,1}^{(t)}]\!], \cdots, [\![g_{i,\tau}^{(t)}]\!]\}$ are regularized by invoking the verification function $\mathbf{ESecJudge}([\![g_i^{(t)}]\!])$. Only those passing the verification are considered \emph{trusted} and included in the index set $\mathcal{I}_\text{trusted}$.

  \item \textbf{Poisonous Gradient Baseline Identification.}
  For all trusted gradients $\{[\![g_i^{(t)}]\!]\}_{i \in \mathcal{I}_\text{trusted}}$, the server computes the cosine similarity with the previously aggregated global gradient $[\![g_\Sigma^{(t-1)}]\!]$ as
  $\cos_{i\Sigma} \leftarrow \mathbf{ESecCos}([\![g_i^{(t)}]\!], [\![g_\Sigma^{(t-1)}]\!]), \quad \text{for } i \in \mathcal{I}_\text{trusted}.$
  The gradient having the \emph{lowest} similarity score is selected as the poisonous baseline
  $[\![g_*^{(t)}]\!]$. %\leftarrow \arg\min_{i \in \mathcal{I}_\text{trusted}} \cos_{i\Sigma}.$

  \item \textbf{Byzantine-Tolerant Aggregation.}
  To defend against adversarial behaviors while incentivizing honest participants, the server constructs a dual-scoring aggregation strategy using both confidence and credit scores.

  \begin{itemize}
    \item \textbf{Confidence Computation.}
    For each user $i \in \mathcal{I}_\text{trusted}$, the cosine similarity between their local gradient and the poisonous baseline is computed as
    $$\cos_{i*} \leftarrow -\ \mathbf{ESecCos}([\![g_i^{(t)}]\!], [\![g_*^{(t)}]\!]).$$
    These similarity values are then normalized to obtain confidence scores:
    \[
      co_i \leftarrow \frac{\exp(\cos_{i*} - \max_j \cos_{j*})}{\sum_{j \in \mathcal{I}_\text{trusted}} \exp(\cos_{j*} - \max_k \cos_{k*})}.
    \]

    \item \textbf{Credit Score Update.}
    Each participant¡¯s credit score $cs_i$ is updated with exponential smoothing
    \[
      cs_i \leftarrow \alpha \cdot cs_i + (1 - \alpha) \cdot co_i,
    \]
    where $\alpha \in [0.7, 0.95]$ is a decay coefficient to balance historical consistency and current behavior.

    \item \textbf{Client Aggregation Weights.}
    The clients' aggregation weights are computed based on the product of credit and confidence scores
    \[
      w_i \leftarrow \frac{cs_i \cdot co_i}{\sum_{j \in \mathcal{I}_\text{trusted}} cs_j \cdot co_j}.
    \]

    \item \textbf{Adaptive Client Filtering.}
    To further enhance robustness against adversarial manipulation and reduce the influence of anomalous or overconfident updates, the server performs an optional three-stage client filtering strategy before final aggregation:
    \begin{itemize}
    \item[(i)] \textit{Adaptive Weight Mixture.} %To further suppress adversarial gradients, the server optionally applies an adaptive weight mixture
    The server first interpolates between uniform weighting and the learned robust weights based on the current round index $t$, using a decaying mixture coefficient:
                \[
                  \lambda \leftarrow \max\left(0, 1 - \frac{t - T_\text{warmup}}{T_\text{total} - T_\text{warmup}}\right),
                \]
                \[
                  \tilde{w}_i \leftarrow \lambda \cdot \frac{1}{n} + (1 - \lambda) \cdot w_i,
                \]
                where $\tilde{w}_i$ denotes the adjusted aggregation weight for client $i$, and $n$ is the total number of clients. This allows the server to rely more on uniform weighting in early rounds and gradually transition to confidence-driven weighting as training progresses. Clients with $\tilde{w}_i < \theta$ (e.g., $\theta = 0.05$) are excluded from the aggregation.

   \item[(ii)]  \textit{Weight-Based Outlier Suppression.}
    To mitigate the impact of disproportionately dominant updates, the server identifies the top-$\frac{n}{2}$ largest weights $\{w_i\}_{i \in \mathcal{I}_\text{trusted}}$, sorts them in descending order, and inspects adjacent weight differences. If a sharp drop is detected (i.e., $w_i - w_{i+1} > \delta$, with $\delta = \frac{1}{2n}$), clients ranked above the drop are discarded. This ensures fairness and guards against potential gradient manipulation. The remaining clients form the final participant set $\mathcal{I}_\text{select}$.

    \item[(iii)] \textit{Credit Penalty for Filtered Clients.}
    To enhance long-term robustness, clients filtered out during either the \textit{Normalization Judgment} or \textit{Outlier Suppression} stages are penalized by applying differentiated scaling factors to their credit scores. Specifically, a multiplicative decay $\gamma_1 \in [0, 1]$ is applied to clients excluded by normalization, while an inflation factor $\gamma_2 > 1$ is used for clients removed during outlier suppression. This adjustment modifies each client's trust trajectory in future rounds accordingly.
    \end{itemize}
  \end{itemize}

    Finally, the weighted aggregation of encrypted gradients is computed as
    \[
    [\![g_\Sigma^{(t)}]\!] \leftarrow \sum_{i \in \mathcal{I}_\text{select}} \frac{w_i}{\sum_{j \in \mathcal{I}_\text{select}} w_j} \cdot [\![g_i^{(t)}]\!].
    \]
\end{itemize}

\textit{4) Model Update:}
The server $S_1$ partially decrypts and sends the share $[g_\Sigma^{(t)}]_{s_i} \leftarrow \mathbf{FHE.PartDec}_{sk_{s_i}}([\![g_\Sigma^{(t)}]\!])$ to each participant $U_i$.
Each user completes decryption using its private key $sk_{u_i}$
\[
  g_\Sigma^{(t)} \leftarrow \mathbf{FHE.FullDec}(\{[g_\Sigma^{(t)}]_{s_i}\}, sk_{u_i}),
\]
and updates its local model using
\[
  W^{(t+1)} \leftarrow W^{(t)} - \eta^{(t)} \cdot g_\Sigma^{(t)},
\]
where $\eta^{(t)}$ is the learning rate of round $t$.

Our proposed privacy-preserving defense is designed to be robust against both conventional data poisoning attacks (as examined in \cite{ma2022shieldfl}) and more sophisticated adaptive model poisoning attacks, including the AGR-tailored and AGR-agnostic strategies introduced in \cite{shejwalkar2021manipulating}. To counter these evolving threats, our defense integrates several key mechanisms that jointly enhance resilience:

\begin{itemize}[leftmargin=0.3cm]
\item \textbf{Dual-Scoring Trust Assessment.}
A two-tier scoring system is employed, combining per-round cosine-based \emph{confidence scores} and long-term \emph{credit scores}. This joint quantification reflects both short-term gradient consistency and historical client reliability, effectively preventing adaptive adversaries from persistently imitating benign behaviors to evade detection.

\item \textbf{Exponential Credit Score Smoothing.}
Credit scores are updated via exponential smoothing with decay factor $\alpha \in [0.7, 0.95]$, enabling the system to gradually penalize clients with recurrent suspicious updates while remaining tolerant to occasional noise in benign gradients.

\item \textbf{Soft Confidence Normalization.}
Instead of applying hard thresholding, we utilize softmax-based normalization of negative cosine similarities to compute smooth confidence scores. This soft treatment maintains sensitivity to fine-grained deviations and provides robustness against gradient-level perturbations crafted by adaptive attackers.

\item \textbf{Adaptive Weight Mixing and Thresholding.} An adaptive mixing scheme transitions from uniform weighting to fully trust-weighted aggregation as training progresses, enhancing early-round stability. Clients with extremely low trust weights (below threshold $\theta$) are excluded from aggregation to reduce stealth attack surfaces.

\item \textbf{Outlier Suppression via Drop Detection.}
Our outlier suppression step explicitly targets highly confident yet anomalously extreme updates, which are likely crafted to exploit the statistical patterns of benign gradients for malicious purposes.
This mechanism is particularly critical for defending against adaptive adversaries, such as those employing AGR-tailored strategies. In such attacks, malicious gradients are deliberately constructed based on the distributional characteristics of benign updates, typically by adding a scaled deviation vector to their mean. These adversarial updates often evade baseline filters due to their close proximity to the benign gradient space. Furthermore, they may achieve abnormally high confidence scores by aligning even more closely with the directionality of benign updates than the benign clients themselves.

\item \textbf{Credit Penalty for Filtered Clients.}
    This mechanism addresses the risk of adaptive adversaries regaining trust too easily after being filtered. Specifically, clients removed during the \textit{Normalization Judgment} step are presumed to have produced non-normalized or low-quality gradients and thus receive a credit decay to reduce their future influence. In contrast, clients filtered via the \textit{Outlier Suppression} mechanism may have presented gradients that appear highly confident yet abnormally extreme. These gradients often arise from adversarial strategies that construct malicious updates by scaling deviations from the benign gradient mean, making them closely aligned with benign directionality and capable of evading baseline checks. Since such clients may initially appear trustworthy, a reverse penalty (i.e., increasing their credit score using $\gamma_2 > 1$) makes them more likely to be flagged again in subsequent outlier filtering, reinforcing detection consistency across rounds.
\end{itemize}

\section{Theoretical Analysis}\label{sec:theorem}
In this section, we delve into an analysis of both the security and complexity of the enhanced scheme.

\subsection{Security Analysis}\label{ssec:secProof}

We prove the security of the enhanced scheme based on the security Definition~\ref{def:new_security}.
\begin{definition}\label{def:new_security}
  Considering a game between an adversary $\mathcal{A}$ and a PPT simulator $\mathcal{S}$,
  if for any real view $\texttt{REAL}$ of $\mathcal{A}$ and the ideal view \texttt{IDEAL} of $\mathcal{S}$, it holds that
  \begin{equation*}
  %\footnotesize
  \scriptsize
    \texttt{REAL}_{\mathcal{A}}^{\prod}(\{g_i^{(t)}\}_{i\in[1,n]},sum, \cos)  \overset{\textit{c}}{\equiv} \texttt{IDEAL}_{\mathcal{S}}^{\prod}(\{g_i^{(t)}\}_{i\in[n]},sum, \cos),
  \end{equation*}
  then the protocol $\prod$ is secure. $\{g_i^{(t)}\}_{i\in[1,n]}$ represents all $n$ users' submitted gradients in the $t$th training iteration.
\end{definition}

\begin{theorem}
  The enhanced scheme can ensure data confidentiality of its privacy-preserving defense process between the two non-colluding and honest-but-curious servers $S_1$ and $S_2$ over all $n$ users' encrypted local gradients against a semi-honest adversary $\mathcal{A}$.
\end{theorem}

\begin{proof}
  We use a standard hybrid argument to prove it.
  For any server $S_1$ and $S_2$, $\texttt{REAL}_{\mathcal{A}}^{\prod}$ includes all users' encrypted gradients $\{[\![g_i]\!]\}_{i\in[1,n]}$ and intermediate parameters $sum$ and $\cos$ of the  \textit{Privacy-Preserving Defense} protocol $\prod$.
  A PPT simulator $\mathcal{S}$ analyzes each implementation procedure in the protocol, %\textit{Privacy-Preserving Defense} process,
  i.e., $\mathcal{S}$ comprises the implementations of normalization judgment, poisonous gradient baseline finding, and Byzantine-tolerance aggregation.
  During these procedures, $\mathcal{S}$ makes a series of (polynomially many) subsequent modifications to the random variable \texttt{REAL}, and we prove that any two subsequent random variables are computationally indistinguishable.
  The detailed proof is given as follows.

  $\mathbf{Hyb_0}$
  This random variable is initialized to be indistinguishable from $\texttt{REAL}_{\mathcal{A}}^{\prod}$ in a real execution process of the protocol.

  $\mathbf{Hyb_1}$ In this hybrid, the behavior of all simulated users $\{U_i\}_{i\in[1,n]}$ are changed that each user $U_i$ encrypts its chosen random gradient vector $\tilde{g}_i^{(t)}\leftarrow \chi_g$ with a public key $pk$ randomly selected by $\mathcal{S}$ to replace the original encrypted gradient vector $[\![g_i^{(t)}]\!]$.
   Since the two-trapdoor FHE is secure under semi-honest adversaries, the two non-colluding honest-but-curious servers cannot distinguish the views between $\{[\![\tilde{g}_i^{(t)}]\!]\}_{i\in[1,n]}$ and $\{[\![g_i^{(t)}]\!]\}_{i\in[1,n]}$.

   $\mathbf{Hyb_2}$ In this hybrid, the product of two ciphertexts $[\![g_i^2]\!]$ in $\mathbf{ESecJudge}$ and $[\![g_a]\!]\cdot[\![g_b]\!]$ in $\mathbf{ESecCos}$ are substituted with random variables $[\![\chi_g^2]\!]$ and $[\![\chi_g]\!]\cdot[\![\chi_g]\!]$, respectively.
   The security of two-trapdoor FHE and the property of multiplicative homomorphism ensure that this hybrid is indistinguishable from the previous one.

  $\mathbf{Hyb_3}$ In this hybrid, all masked gradients in the privacy-preserving defense process are substituted.
  Specifically, we substitute the masked gradients $[\![g_{i,j}^2+r_j]\!]_{i\in[1,n],j\in[1,\tau]}$ during the executions of $\mathbf{ESecJudge}([\![g_{i}]\!])_{i\in[1,n]}$ in the normalization judgment procedure (resp. the masked gradients $[\![g_{i,j}^{(t)}\cdot g_{\Sigma,j}^{(t-1)}+r_j]\!]_{i\in[1,n],j\in[1,\tau]}$ during the executions of $\mathbf{ESecCos}([\![g_i^{(t)}]\!],[\![g_{\Sigma}^{(t-1)}]\!])_{i\in[1,n]}$ in the poisonous gradient baseline finding procedure and the masked gradients $[\![g_{i,j}^{(t)}\cdot g_{*,j}^{(t)}+r_j]\!]_{i\in[1,n],j\in[1,\tau]}$ during the executions of $\mathbf{ESecCos}([\![g_i^{(t)}]\!],[\![g_{*}^{(t)}]\!])_{i\in[1,n]}$ in the Byzantine-tolerance aggregation procedure)  with random variables $[\![\chi_{i,j}]\!]_{i\in[1,n],j\in[1,\tau]}$.
  The semi-honest security of the two-trapdoor FHE and the non-colluding setting between the two servers guarantee that this hybrid is indistinguishable from the previous one.

  $\mathbf{Hyb_4}$ In this hybrid, the product of two ciphertexts $[\![g_i^2]\!]$ in $\mathbf{ESecJudge}$ and $[\![g_a]\!]\cdot[\![g_b]\!]$ in $\mathbf{ESecCos}$ are separately masked by random polynomial $r_i\leftarrow\mathcal{R}_p$, which can be perfectly simulated by $\mathcal{S}$. 
  The intermediate parameters, also masked by uniformly random polynomials,  maintain a uniform distribution.
  As addition homomorphism ensures that the distribution of any intermediate variables remains consistent, $\mathcal{A}$ cannot distinguish the output of $\mathcal{S}$ in polynomial time, thus ensuring that this hybrid is identically distributed to the previous one.

  $\mathbf{Hyb_5}$ In this hybrid, the aggregated gradient $[\![g_{\Sigma}^{(t)}]\!]$ in the Byzantine-tolerance aggregation procedure is calculated using addition homomorphism. The real view of $\mathcal{S}$ is $\texttt{Ideal}_{\mathcal{S}}^{\prod}=\{\{g_i{(t)}\}_{i\in[1,n]},$ $sum,\cos,[\![g_{\Sigma}^{(t)}]\!]\}$, where $sum$ and $\cos$ are sets comprising the sums and cosine similarities obtained when executing $\{\mathbf{ESecJudge}([\![g_i^{(t)}]\!])\}_{i\in[1,n]}$ and $\{\mathbf{ESecCos}([\![g_i^{(t)}]\!],[\![g_{\Sigma}^{(t-1)}]\!])\}_{i\in[1,n]}$, respectively. Since the inputs $\{[\![g_i^{(t)}]\!]\}_{i\in[1,n]}$ are all encrypted, and the intermediate results are masked by random values, the non-colluding servers involved in calculating $sum$ and $\cos$ cannot compromise participants' privacy without access to any information about the inputs and intermediate results. Hence, this hybrid is indistinguishable from the previous one.

  In sum up, we can deduce from the above argument that $\texttt{REAL}_{\mathcal{A}}^{\prod}(\{g_i^{(t)}\}_{i\in[1,n]},sum, \cos)  \overset{\textit{c}}{\equiv} \texttt{IDEAL}_{\mathcal{S}}^{\prod}(\{g_i^{(t)}\}_{i\in[n]},$ $sum, \cos)$ holds, which indicates that our enhanced scheme can ensure data confidentiality.

\end{proof}

\subsection{Convergence Analysis}
We provide a detailed theoretical convergence analysis of the proposed PBFL scheme in Appendix~C. Specifically, we decompose the aggregation error into three components---(i) residual Byzantine bias, (ii) stochastic gradient noise, and (iii) gradient diversity---arising respectively from malicious updates, local sampling noise, and data heterogeneity. We then formally show that, under standard smoothness and bounded-variance assumptions, the proposed algorithm achieves the same sub-linear convergence rate $\mathcal{O}(1/\sqrt{T})$ as classical FedAvg, up to an additive error term that is proportional to the Byzantine fraction $\beta$ and decreases as the number of trusted participants increases. This result theoretically confirms the robustness and efficiency of our defense mechanism.

\subsection{Performance Analysis}\label{ssec:performanceComp}
We provide a detailed performance comparison of computational and communication complexities between our proposed scheme and \nameScheme, which is presented in Appendix~D. The analysis covers all major phases including \textit{Setup}, \textit{Local Training}, \textit{Privacy-Preserving Defense}, and \textit{Model Update}.
Our results show that, due to ciphertext computation simplicity and cryptographic primitive efficiency, our scheme significantly outperforms \nameScheme in both computation and communication efficiency. Specifically:

\begin{itemize}[leftmargin=0.3cm]
    \item \textbf{Computational Complexity:} Our scheme reduces asymptotic encryption, decryption, and homomorphic operation costs by a factor related to the smaller ciphertext dimension (\(n_f q \ll N^2\)). This translates to a lower overhead in all training rounds.
    \item \textbf{Communication Cost:} Since homomorphic multiplication in our secure computation methods $\mathbf{ESecCos}$ and $\mathbf{ESecJudge}$ is executed locally without inter-party interaction, its communication cost is effectively zero. In contrast, \nameScheme incurs significant overhead due to interactive homomorphic multiplication in their secure computation protocols. 
\end{itemize}

%This theoretical performance advantage is also reflected in our experimental results (Section~\ref{ssec:evaluationResults}), validating the practical efficiency of our scheme.

\subsection{Discussion and Practical Considerations}
While our enhanced scheme achieves significant improvements in privacy and efficiency, several practical considerations remain for real-world FL deployments. This section discusses two key aspects: the feasibility of decentralizing the Key Center and the challenges of defending against adaptive and malicious adversaries.

\textbf{(1) Decentralizing the Key Center:}
Our current design assumes a trusted centralized Key Center for key management and distribution. This assumption, however, may not always hold in heterogeneous or cross-silo FL environments where participants belong to different administrative domains with varying trust levels. To mitigate this limitation, decentralized approaches such as threshold cryptography and Distributed Key Generation (DKG) protocols can be adopted. These techniques distribute key management responsibilities among multiple independent parties, thereby eliminating the single point of trust and reducing the risk of key compromise. In addition, blockchain-based solutions and secure multiparty computation (MPC) protocols could be integrated to provide transparent, auditable, and tamper-resistant key management. While these decentralized mechanisms can enhance robustness and trustworthiness, they inevitably introduce challenges such as increased communication overhead and coordination complexity, which require further study.

\textbf{(2) Defending Against Adaptive and Malicious Threats:}
Our security analysis primarily considers malicious clients and at most one corrupted (malicious) server, ensuring the privacy of honest clients' data under the assumption that the two servers do not collude. However, in real-world FL deployments, adversaries may behave both maliciously and adaptively, dynamically deviating from the protocol to launch sophisticated attacks such as adaptive model poisoning or adaptive inference.

We explicitly evaluate our framework under an \emph{adaptive model poisoning} threat, where the adversary dynamically adjusts its malicious gradients based on the gradients of benign clients before injecting them into the training process. This adaptive strategy makes detection substantially more difficult. Our experimental results demonstrate that the proposed framework can effectively mitigate this class of adaptive attacks. Beyond model poisoning, more advanced adaptive threats have been reported, including adaptive backdoor attacks that iteratively modify backdoor triggers to evade defenses and adaptive inference attacks that progressively extract private information from model updates. These attacks evolve over training rounds and present challenges that our current design does not fully address.

Future work will focus on strengthening the framework against these stronger adversarial models by extending robustness beyond poisoning to cover adaptive backdoor and inference-based attacks. Potential directions include developing adaptive anomaly detection frameworks that dynamically adjust detection thresholds or defense strategies based on adversarial behaviors observed during training (e.g., online learning-based detectors that evolve as attacks change) and incorporating temporal consistency checks and cross-round validation mechanisms to detect adversaries that gradually adapt their malicious updates to evade static defenses. Furthermore, leveraging ensemble or diversity-based defenses (where multiple defense mechanisms operate in parallel and dynamically reconfigure based on the detected threat level) could make it substantially harder for adaptive attackers to circumvent all defenses simultaneously. Beyond adaptive attacks, we also plan to address the threat of malicious servers that compromise computation correctness by integrating verifiable computation or zero-knowledge proofs to guarantee both privacy and correctness even when a server is malicious.

%% file: sections/6-conclusion.tex
%\vspace*{-1\baselineskip}
\section{Conclusion}\label{sec:conclusion}

In this paper, we identify security weaknesses in privacy-preserving and Byzantine-robust federated learning (PBFL) schemes with the two-server model and proposes an enhanced PBFL solution.
In our enhanced scheme, a Byzantine-tolerant aggregation method is devised to fortify against model poisoning attacks. Additionally, we develop the enhanced secure normalization judgment method and secure cosine similarity measurement method to protect the privacy of the defense process while ensuring the correctness of the computations.
Our scheme guarantees privacy preservation and resilience against model poisoning attacks, even in scenarios with heterogeneous datasets that are non-IID (Independently Identically Distributed).
Moreover, it demonstrates superior efficiency in computational and communication performance compared to the existing state-of-the-art PBFL schemes.